\documentclass{article}
\usepackage{ijcai24}

\usepackage{times}
\usepackage{soul}
\usepackage{url}
\usepackage[hidelinks]{hyperref}
\usepackage[utf8]{inputenc}
\usepackage[small]{caption}
\usepackage{graphicx}
\usepackage{amsmath}
\usepackage{amsthm}
\usepackage{booktabs}
\usepackage{algorithm}
\usepackage{algorithmic}
\usepackage[switch]{lineno}

\usepackage{import}
\usepackage{preamble}

\title{Game Transformations That Preserve Nash Equilibria or Best-Response Sets\footnote{Published in the Proceedings of the Thirty-Third International Joint Conference on Artificial Intelligence (IJCAI 2024). An earlier extended abstract of this paper can be found in the Proceedings of the 23rd International Conference on Autonomous Agents and Multiagent Systems (AAMAS 2024).}}

\author{
Emanuel Tewolde
\And
Vincent Conitzer
\affiliations
Foundations of Cooperative AI Lab (FOCAL)\\
Computer Science Department, Carnegie Mellon University, USA\\
\emails
emanueltewolde@cmu.edu,
conitzer@cs.cmu.edu}

\begin{document}

\maketitle

\begin{abstract}
In this paper, we investigate under which conditions normal-form games are (guaranteed to be) strategically equivalent. First, we show for $N$-player games ($N \geq 3$) that \begin{enumerate}[leftmargin=0.6cm]
\item[(A)] it is NP-hard to decide whether a given strategy is a best response to some strategy profile of the opponents, and that 
\item[(B)] it is co-NP-hard to decide whether two games have the same best-response sets.
\end{enumerate}
Combining that with known results from the literature, we move our attention to equivalence-preserving game transformations.

It is a widely used fact that a positive affine (linear) transformation of the utility payoffs neither changes the best-response sets nor the \NE{} set. We investigate which other game transformations also possess either of the following two properties when being applied to an arbitrary $N$-player game ($N \geq 2$): 
\begin{enumerate}[label=(\roman*), leftmargin=0.6cm]
\item The \NE{} set stays the same;
\item The best-response sets stay the same. 
\end{enumerate}

For game transformations that operate player-wise and strategy-wise, we prove that (i) implies (ii) and that transformations with property (ii) must be positive affine. The resulting equivalence chain highlights the special status of positive affine transformations among all the transformation procedures that preserve key game-theoretic characteristics.

\end{abstract}

\section{Introduction}
\label{sec:introduction}

\subsection{Motivation} 
\label{sec:motivation}
When faced with a strategic interaction with other agents, it can be computationally useful for AI systems -- as we will discuss further down -- to detect when the current situation can be treated in the same way as another strategic game that has already been dealt with in the past. This problem can also be critical for the robustness and generalizability of our AI systems. AI agents will often need to act in new environments, possibly with multiple equilibria, and yet be predictable to each other and to us to avoid bad outcomes. Recognizing a new environment as strategically equivalent to a previously encountered one can help tremendously with this, providing a precedent for action that ensures that everyone’s expectations on behavior are well calibrated. \citet{OesterheldC22}, for example, build on that in order to achieve Pareto-improved outcomes in games played by AI representatives (where games can be transformed by reprogramming the AIs' rewards).

Not least for these various reasons, even the simplest class of games -- $2$-player normal-form with $2$ actions per player -- has been studied extensively in order to obtain a complete taxonomy for them; see for example \cite{RobinsonG05,RapoportGG76,Borm87}. With it, it is easy to recognize when a $2 \times 2$ game contains traits of competition, cooperation, coordination, etc.~\cite{BrunsK21}. Such interpretation tools are also being developed for more complex strategic situations \cite{marris23}, but this still remains an important venue of further research. One challenge is that larger games become prohibitively complex to compare directly: 
\citet{10.1016/j.tcs.2008.07.021} show that deciding whether two $2$-player normal-form games share \NEs{} is a computationally hard task. We will show in Section~\ref{sec:decaboutbrs} that this task is also computationally hard for the case of best response sets and at least three players.

One classic tool that emerged in the beginnings of game theory has been to transform a given game into other strategically equivalent games that are easier to analyze \cite{10.2307/j.ctt1r2gkx}. 
Positive affine (linear) transformations (PATs) have been particularly useful in that regard \cite{Aumann1961AlmostSC,10.1007/978-3-642-10841-9_44}. To illustrate PATs, consider any $2$-player normal-form game in which the players' utilities are measured in dollars. Then, the best-response strategies of player~1 do not change if her utility payoffs are multiplied by a factor of~$5$. Moreover, they also do not change if $10$ dollars are added to all outcomes that involve player 2 playing his, say, third strategy. More generally, PATs have the power to rescale the utility payoffs of each player and to add constant terms to the utility payoffs of a player $i$ for each strategy choice $\ks_{-i}$ of her opponents. 

Through leveraging PATs, previous work significantly extended the applicability of efficient \NE{} solvers \cite{Neumann1928,Dantzig51,Adler2013,main} to classes beyond those of zero-sum and rank-$1$ games\footnote{A $2$-player game, represented by its payoff matrices $A,B \in \R^{m \times n}$, is said to have rank $1$ if $\rank(A+B) = 1$.} \cite{moulin,KONTOGIANNIS201264,HeymanG23}.

PATs are also popular in mechanism design and e-commerce: \emph{Affine maximizer auctions} are PAT transformations of the classic VCG mechanism, and as such, inherit strategy-proofness and individual rationality by the strategy-preserving nature of PATs. They play a key role in finding revenue-maximizing mechanisms (both with classical optimizers \cite{LikhodedovS04} and deep learning \cite{CurrySD23,CurryTCMKSHS24}, and in improving welfare in redistribution mechanisms \cite{GuoC10} and advertisement auctions \cite{DengMMZ21}.

The versatility of PATs is based on their well-known property that they do not change preferences, best responses, or Nash equilibria, when being applied to an arbitrary game. In a very precise sense, PATs are also \emph{the only} game transformations that do not change preferences; cf.\ Section~\ref{sec:literature review}. The main result in this paper addresses the question of whether there are other (efficiently computable) game transformations that do not change best responses or Nash equilibria.

\subsection{Overview}
\label{sec:overview}

Sections~\ref{sec:nfgames} and~\ref{sec:gametrafos} provide some background on game-theoretic concepts that are relevant to understanding and deriving our main results. In Section~\ref{sec:decaboutbrs}, we develop computational hardness results for deciding whether a strategy in a game ever constitutes a best response and for deciding whether two games have the same best-response sets. We believe these results are of independent interest. However, they are also important for Section~\ref{sec:discussion}, in which we discuss why we will henceforth restrict our attention to game transformations that transform utilities player-wise and strategy-wise (called \textit{separability}). In Section~\ref{sec:strat equiv preserving}, we proceed to characterize all separable game transformations that preserve the \NE{} set -- or, alternatively, the best response sets -- when being applied to an arbitrary $N$-player game. Last but not least, Section~\ref{sec:literature review} puts our results into context with further related work. 

To illustrate the insights of Section~\ref{sec:strat equiv preserving} on an example, consider~$H_{\textnormal{Ex}}$ that takes any $2$-player $2 \times 2$ normal-form game with payoff matrices
\[A = \begin{pmatrix}
a_{11}  & a_{12}\\
a_{21} & a_{22}
\end{pmatrix} \quad , \quad
B = \begin{pmatrix}
b_{11} & b_{12}\\
b_{21} & b_{22}
\end{pmatrix} \]
and transforms it into the game $H_{\textnormal{Ex}}(A,B) := (A', B')$ that is defined as
\[A' = \begin{pmatrix}
-2 a_{11} +10  & a_{12}^5\\
e^{a_{21}} & 0
\end{pmatrix} 
\, , \, B'= \begin{pmatrix}
|b_{11}| & \sign (b_{12})\\
\sqrt{|b_{21}|} & \arctan (b_{22})
\end{pmatrix} \]
As one can see with the sign function in $B'$, it is noteworthy to highlight that our notion of a game transformation allows for non-continuous functions. With Theorem~\ref{equiv charact of PAT}, we will show that there must exist $2 \times 2$ games $(\bar{A}, \bar{B})$ for which $H_{\textnormal{Ex}}$ does not preserve their Nash equilibrium set or - respectively -  their best-response sets. More generally, we derive that \textit{universally} preserving the Nash equilibrium set implies that the best-response sets always have to be preserved as well; and that the latter property is only satisfied by game transformations~$H$ with the very restricted structure of a PAT. In the example of $H_{\textnormal{Ex}}$, each transformation map within it single-handedly already violates a PAT structure.

Full proofs for statements in this paper can be found in the appendix.

\section{Normal-Form Games}
\label{sec:nfgames}

Notation-wise, we denote $[n] := \{1,\ldots,n\}$ for any $n \in \N$.
\\
A normal-form multiplayer game $G$ specifies 
\begin{enumerate}
    \item[(a)] the number of players $N \in \N, N \geq 2$,
    \item[(b)] a set of pure strategies $S^i = [m_i]$ for each player~$i$ where $m_i \in \N$, $m_i \geq 2$, and
    \item[(c)] the utility payoffs for each player~$i$ given as a function $u_i: S^1 \times \ldots \times S^N \rightarrow \R$.
\end{enumerate}
Denote the set of strategy profiles in $G$ as $S := S^1 \times \ldots \times S^N$. Throughout this paper, all considered multiplayer games shall have the same number of players $N$ and the same set of strategy profiles~$S$. Hence, any game $G$ will be determined by its utility functions $\{u_i\}_{i \in [N]}$. The players choose their strategies simultaneously and they cannot communicate with each other. A utility function $u_i$ can be summarized by its pure strategy outcomes for player~$i$, captured as an $N$-dimensional tensor or array $\big\{ u_i(\ks) \big\}_{\ks \in S}$.

As usual, we allow the players to randomize over their pure strategies, called mixed strategies. Then, player~$i$'s strategy space extends to the set of probability distributions 
    $\Delta(S^i) := \, \big\{ s^i = (s_k^i)_k \in \R_{\geq 0}^{m_i} : \,$
    $\sum_{k \in [m_i]} s_k^i = 1 \big\}$
over $S^i$. A tuple
    $\strats = (s^1, \ldots, s^N) \in \Delta(S^1) \times \ldots \times \Delta(S^N) =: \Delta(S)$
is called a mixed strategy profile\footnote{Not to be confused with a correlated strategy: In our notation, $\Delta(S)$ itself is not a simplex of high dimension but only the product of $N$ lower-dimensional simplices.} in $G$. The utility payoff of player~$i$ under profile $\strats$ is defined as the player's utility payoff in expectation 
$u_i(\strats) := \sum_{\ks \in S} s_{k_1}^1 \cdot \ldots \cdot s_{k_N}^N \cdot u_i(\ks)$.
The goal of each player is to maximize her utility.

We will abbreviate with $S^{-i}$ the set that consists of all possible pure strategy choices $\ks_{-i} = (k_1, \ldots, k_{i-1},k_{i+1}, \ldots, k_N)$ of the opponent players (resp. $\Delta(S^{-i})$ for the set of mixed strategy choices $\strats^{-i} = (s^1, \ldots, s^{i-1},s^{i+1}, \ldots, s^N)$). We will also use $u_i(k_i,\ks_{-i})$ instead of $u_i(\ks)$ to stress how player~$i$ can only influence her own strategy when it comes to her payoff (resp. $u_i(s^i,\strats^{-i})$ instead of $u_i(\strats)$).
\begin{defn}
The best-response set of player~$i$ to the opponents' strategy choices $\strats^{-i}$ is defined as 
\\
    $\BR_{u_i}(\strats^{-i}) :=  \argmax_{t^i \in \Delta(S^i)} \big\{ \, u_i(t^i,\strats^{-i}) \, \big\}$.
\end{defn} 

Best-response strategies capture the idea of optimal play against the other player's strategy choices. The most popular equilibrium concept in non-cooperative games is based on best responses.
\begin{defn}
A strategy profile $\strats \in \Delta(S)$ to a game $G = \{u_i\}_{i \in [N]}$ is called a \NE{} if for every player~$i \in [N]$ we have $s^i \in \BR_{u_i}(\strats^{-i})$.
\end{defn}
\noindent
By a result of \citet{Nash48}, any such multiplayer game $G$ admits at least one \NE{}.

\section{Decision Problems about Best Responses}
\label{sec:decaboutbrs}

In this section we show that two decision problems about best responses are hard for $N$-player games, when $N \geq 3$. To our knowledge, these results are novel.

For computational problems involving $N$-player games $G$ with strategy sets $(S^i)_{i \in [N]}$ and utility functions $(u_i)_{i \in [N]}$, we are interested in their computational complexities in terms of $|S|$ and the binary encoding of all utility payoffs $( u_i(\strats) )_{\strats \in S, i \in [N]}$. For that, we require that utility payoffs take on rational values only. 

First, we consider the problem of deciding whether a mixed strategy of a player is ever a best response to some mixed strategy profile of the opponent players. In its computationally easiest form, we may formulate it as the following.
\begin{defn}[{\sc CheckIfEverBR}]
    Given a $3$-player normal-form game, do there exist mixed strategies $\rs \in \Delta(S^2)$ of \textnormal{PL2} and $\strats \in \Delta(S^3)$ of \textnormal{PL3} such that pure strategy $1$ of \textnormal{PL1} is a best response to $(\rs, \strats)$?
\end{defn}
This is different from determining the best responses of a player to a given strategy profile of the opponents, a task that can be solved in polynomial time. Our problem is related to \textit{rationalizable} strategies \cite{Bernheim84,Pearce84} - a concept that is based on the idea that a rational player can and should eliminate any strategy that is not a best response to some belief over what her opponents may play.
\begin{prop}
\label{checkifeverbr}
    {\sc CheckIfEverBR} is $\mathbf{NP}$-hard.
\end{prop}
The analogous formulation of {\sc CheckIfEverBR} for the case of $2$-player games can be efficiently decided by solving a system of linear (in-)equalities. We can recover polynomial-time solvability for many-player games if we allow the opponents to play in a coordinated fashion (cf.~correlated strategies). On a related note, \citet{Pearce84}[Lemma 3] shows that a strategy $s^*$ is a best response to some correlated strategy of the opponents if and only if $s^*$ is not strictly dominated by a mixed strategy.

We prove Proposition~\ref{checkifeverbr} by a reduction from the Balanced Complete Bipartite Subgraph problem. This decision problem asks whether a given weighted bipartite graph $G = (V \cup W,E)$ has subsets $V^* \subseteq V$ and $W^* \subseteq W$ of given size $K \in \N$ that are fully connected, that is, $(v,w) \in E$ for all $v \in V^*, w \in W^*$. This problem is known to be $\mathbf{NP}$-complete \cite{Garey90}[GT24].
\begin{proof}[Proof sketch of Proposition~\ref{checkifeverbr}]
Given an instance $G = (V \cup W,E)$ and $K$ of the Balanced Complete Bipartite Subgraph problem, construct a three player game where \textnormal{PL2} has strategy set $V$ and \textnormal{PL3} has strategy set $W$. \textnormal{PL1} will have the following strategies: Strategy ``$1$'' which will be the subject of interest in {\sc CheckIfEverBR}, one strategy for each node in $G$, and one strategy for each edge $(v,w) \in V \times W$ that is not present in $G$. The utility payoffs of \textnormal{PL1} will be carefully constructed such that strategy $1$ is a best response to mixed strategies $(\rs, \strats)$ of \textnormal{PL2} and \textnormal{PL3} if and only if the support of $\rs$ and $\strats$ form subsets $V^*$ and $W^*$ that make a balanced complete bipartite subgraph of $G$. To that end, we make strategy $v$ (resp. $w$) of \textnormal{PL1} very attractive for \textnormal{PL1} in the case that \textnormal{PL2} (resp. \textnormal{PL3}) plays their corresponding strategy $v$ (resp. $w$) with too much probability. Moreover, we make a strategy $(v,w) \notin E$ of \textnormal{PL1} very attractive for \textnormal{PL1} in the case that \textnormal{PL2} and \textnormal{PL3} both play their corresponding strategies $v$ and $w$ with any significant probability at all. Intuitively, these two conditions accomplish that in any potential certificate~$(\rs, \strats)$, \textnormal{PL2} and \textnormal{PL3} will mix over at least $K$ strategies and, moreover, they will only put non-negligible weight on strategies $v$ and $w$ if $(v,w) \in E$.
\end{proof}

Based on the hardness of {\sc CheckIfEverBR}, we can prove $\mathbf{co\text{-}NP}$-hardness of deciding best-response equivalence.
\begin{defn}[{\sc CheckIfSameBRs}]
    Given two $3$-player normal-form games with strategy set $S^1 \times S^2 \times S^3$, do they have the same best-response sets?
\end{defn}
\begin{thm}
\label{checkifequalBRs}
    {\sc CheckIfSameBRs} is $\mathbf{co\text{-}NP}$-hard.
\end{thm}
\begin{proof}[Proof sketch]
Given a game instance $G$ of {\sc CheckIfEverBR}, construct another game $G'$ by changing the utility that \textnormal{PL1} receives from playing strategy $1$ to something worse than the lowest payoff present in $G$. If a best-response set changed from $G$ to $G'$, then it must also be the case that strategy $1$ for \textnormal{PL1} was added or removed from that best-response set. The former cannot happen because strategy $1$ is strictly dominated for \textnormal{PL1} in $G'$ which prevents it from ever being a best response. Thus, $G$ and $G'$ will have the same best-response sets if and only if strategy $1$ is never a best-response strategy in $G$.
\end{proof}
Together with prior work found in the literature, Theorem~\ref{checkifequalBRs} will guide us in the next sections when it comes to the types of game transformations that we may consider for preserving key game-theoretic characteristics. We believe, however, that Proposition~\ref{checkifeverbr} and Theorem~\ref{checkifequalBRs} are also of independent interest for algorithmic game theory and AI research.

\section{Preliminaries on Game Transformations}
\label{sec:gametrafos} 

\subsection{Positive Affine Transformations} 

The following lemma (or restricted versions of it) is a well-known result for $2$-player games.\footnote{ See \citet{HeymanG23}[Lemma 2.1], \citet{maschler_solan_zamir_2013}[Theorem 5.35], \cite{moulin}[Theorem 1], \citet{harsanyi1988general}[Chapter 3] or \citet{DynGT}[Proposition 3.1].} Here, the notation $\1_n \in \R^n$ stands for the vector with all ones as its entries.
\begin{lemma}
\label{PAT preserves lemma}
Let $(A,B)$ be an $m_1 \times m_2$ bimatrix game and take arbitrary scalars $\alpha_1, \alpha_2 >0$ and vectors $c^1 \in \R^{m_2}, c^2 \in \R^{m_1}$. Define 
\[
    A' = \alpha_1 A + \1_{m_1} (c^1)^T \quad \textnormal{and} \quad B' = \alpha_2 B + c^2 \1_{m_2}^T \, \, .
\]
Then $(A', B')$ has the same best-response sets as $(A,B)$. 
Thus, both games have the same \NE{} set.
\end{lemma}
The game transformations in Lemma \ref{PAT preserves lemma} are called ($2$-player) positive affine transformations (PATs). An explicit example of a $2$-player PAT is one that transforms a $2 \times 2$ game $(A,B)$ into
\[A' = \begin{pmatrix}
2 a_{11} + 10  & 2a_{12} - 5\\
2 a_{21} + 10 & 2a_{22} - 5
\end{pmatrix} \, , \]
\[ B'= \begin{pmatrix}
\frac{1}{2}b_{11} & \frac{1}{2}b_{12} \\
\frac{1}{2}b_{21} -\sqrt{3} & \frac{1}{2}b_{22} -\sqrt{3}
\end{pmatrix} \, . \]

The intuition behind Lemma~\ref{PAT preserves lemma} is as follows: \textnormal{PL1} wants to maximize her utility given the strategy of \textnormal{PL2}. A positive rescaling of $u_1$ will change the utility payoffs but not the utility-maximizing strategies. The same holds true if we add utility payoffs to $u_1$ that are only dependent on the strategy choice of her opponent \textnormal{PL2}, because that would make a constant shift in terms of the decision variables of \textnormal{PL1}. 
\\
Let us generalize PATs to multiplayer games.
\begin{defn}
\label{multiplayer PAT defn}
A positive affine transformation (PAT) specifies for each player~$i$ a scaling parameter $\alpha^i \in \R, \alpha^i >0,$ and translation constants $C^i := ( c_{\ks_{-i}}^i)_{\ks_{-i} \in S^{-i}}$ for each choice of pure strategies from the opponents. 
The PAT $H_{\textnormal{PAT}} = \big\{ \alpha^i, C^i \big\}_{i \in [N]}$ then takes any game $G = \{u_i\}_{i \in [N]}$ as an input and returns the transformed game $H_{\textnormal{PAT}}(G) = \{u_i'\}_{i \in [N]}$ with utility functions
\begin{align}
\label{PAT transformed utilities}
\begin{aligned}
u_i' : S \rightarrow \R
\, \, , \, \,
\ks \mapsto \alpha_i \cdot u_i(\ks) + c_{\ks_{-i}}^i \, .
\end{aligned}
\end{align}
\end{defn}
Multiplayer PATs also preserve the best-response sets and \NE{} set, which we prove in Appendix~\ref{sec:helpinglemmas} for completeness.
\begin{lemma}
\label{multiplayer PAT preserves}
Take a PAT $H_{\textnormal{PAT}} = \big\{ \alpha^i, C^i \big\}_{i \in [N]}$ and any game $G = \{u_i\}_{i \in [N]}$. Then, the transformed game $H_{\textnormal{PAT}}(G) = \{u_i'\}_{i \in [N]}$ has the same best-response sets as the original game $G$. Consequently, $H_{\textnormal{PAT}}(G)$ also has the same \NE{} set as $G$.
\end{lemma}
PATs have found much success as a tool for simplifying a given game precisely because of this property. We want to investigate which other game transformations also preserve the best-response sets or the \NE{} set. If we found more of these transformations, we could use them to, e.g., further increase the class of efficiently solvable games.

\subsection{Separable Game Transformations}

In this paper, we will focus on the following space of game transformations. We discuss in Section~\ref{sec:discussion} why this forms a maximally large search space within which we may still reasonably hope to find game transformation that are equivalence-preserving and efficiently computable.
\begin{defn}
\label{def game trafo}
A separable game transformation $H = \{H^i\}_{i \in [N]}$ specifies for each player~$i$ a collection of functions 
$H^i := \big\{ h_{\ks}^i : \R \rightarrow \R \big\}_{\ks \in S}$, 
indexed by the different pure strategy profiles $\ks$. \\
The transformation $H$ can then be applied to any $N$-player game $G = \{u_i\}_{i \in [N]}$ with strategy set $S$ to construct the transformed game $H(G) = \{H^i(u_i)\}_{i \in [N]}$ where 
    \\
    $H^i(u_i) : S \to \R, \quad \ks \mapsto h_{\ks}^i \big( u_i(\ks) \big)$.
\end{defn}
Observe that the utility payoff of player~$i$ in the transformed game $H(G)$ from the pure strategy outcome $\ks$ is only a function of the utility payoff from that same player in that same pure strategy outcome of the original game~$G$.

We extend the utility functions $H^i(u_i)$ to mixed strategy profiles $\strats \in \Delta(S)$ as usual through 
\\
    $H^i(u_i)(\strats) := \sum_{\ks \in S} s_{k_1}^1 \cdot \ldots \cdot s_{k_N}^N \cdot h_{\ks}^i \big( u_i(\ks) \big)$.
\\
To simplify future notation, we will often use $h_{k_i,\ks_{-i}}^i$ to refer to $h_{\ks}^i$.

\begin{rem}
A multiplayer positive affine transformation $H_{\textnormal{PAT}} = \big\{ \alpha^i, C^i \big\}_{i \in [N]}$ makes a separable game transformation $H = \{H^i\}_{i \in [N]}$ by setting 
\\
$h_{\ks}^i : \, \, \R \to \R$
\, , \,
$z \mapsto \alpha^i  \cdot z + c_{\ks_{-i}}^i$.
\end{rem}
In the following Definitions~\ref{defn NE preserving} and~\ref{defn BR preserving}, we define the universally preserving characteristics that we are interested in.
\begin{defn}
\label{defn NE preserving}
Let $H = \{H^i\}_{i \in [N]}$ be a separable game transformation. Then we say that $H$ universally preserves \NE{} sets if for all games $G = \{u_i\}_{i \in [N]}$ the transformed game $H(G) = \{H^i(u_i)\}_{i \in [N]}$ has the same \NE{} set as~$G$.
\end{defn}

\begin{defn}
\label{defn BR preserving}
Let map $H^i$ come from a separable game transformation $H$. Then we say that $H^i$ universally preserves best responses if for all utility functions $u_i : S \rightarrow \R$ and for all opponents' strategy choices $\strats^{-i} \in \Delta(S^{-i})$:
\begin{align*}
\BR_{H^i(u_i)}(\strats^{-i}) &= \argmax_{t^i \in \Delta(S^i)} \big\{ H^i(u_i)(t^i,\strats^{-i}) \big\} 
\\
&= \argmax_{t^i \in \Delta(S^i)} \big\{ u_i(t^i,\strats^{-i}) \big\} = \BR_{u_i}(\strats^{-i}) \, .
\end{align*}
\end{defn}

Lemma~\ref{multiplayer PAT preserves} states that the maps $H^i$ of a PAT universally preserve best responses. Note, moreover, that by definition of a \NE{}, a game transformation $H = \{H^i\}_{i \in [N]}$ will universally preserve \NE{} sets if for every player~$i$ the map $H^i$ universally preserves best responses. Therefore, being a PAT implies Definition~\ref{defn BR preserving} implies Definition~\ref{defn NE preserving}. In Section~\ref{sec:strat equiv preserving} we will show the reverse implication chain for game transformations that are separable.

\section{Discussion of Restrictions}
\label{sec:discussion}

The space of separable game transformations forms a vast landscape in which we may search for universally preserving transformations. This can be seen from the game transformation example $H_{\textnormal{Ex}}$ of Section~\ref{sec:overview}. However, one might still ask why this paper does not expand its attention to non-separable game transformations. We will discuss that in this section.

For example, consider a game transformation that introduces or removes duplicate strategies or dummy players. Note that this would require the transformations to have the power to change the strategy sets and player set. Nonetheless, these specific examples are well-behaved in the sense that they alter the \NE{} set (or best responses) in an easily describable manner. \citet{AbdouPSV22}, for example, managed to characterize how selected examples of these transformations interact with various methods of decomposing a game. Transformations that change the strategy sets and player set also appear in the literature under the term \textit{Nash homomorphism}, and they have been of use for complexity-theoretic studies, e.g., of win-lose games \cite{Abbott05onthe} or ranking games \cite{10.5555/1597538.1597636}. Suffice to say, once we allow for game transformations to arbitrarily change the game structure, i.e. the player set and strategy sets, it is not straightforward to define anymore under what conditions two games of different game structure should be considered ``strategically equivalent''. This makes such general game transformations prohibitively complex (or impossible) to analyze beyond a case by case basis. Therefore, and in accordance with most of the literature on strategic equivalence between games \cite{moulin,MORRIS2004260,10.1016/j.tcs.2008.07.021,Liu96}, we restrict our attention to games whose game structures are directly comparable.

Indeed, game transformations that preserve the player set and the strategy sets form an interesting search space because Definitions~\ref{defn NE preserving} and~\ref{defn BR preserving} can be directly extended to it and because within that search space, some of our following results will not hold true anymore. Compare the Prisoner's Dilemma with the Quality game, as presented by \citet{vonStengel21}:
\[ \begin{pmatrix}
2,2  & 0,3\\
3,0 & 1,1
\end{pmatrix}
\textnormal{ and }
\begin{pmatrix}
2,2 & 0,1\\
3,0 & 1,1
\end{pmatrix} \, . \]
Both games have 
the same unique \NE{}, namely, where \textnormal{PL1} plays the bottom row and \textnormal{PL2} plays the right column. But the best response of \textnormal{PL2} to \textnormal{PL1} playing the top row is different in the two games. This example illustrates the fact that strictly dominated strategies will never be a best response, and so they will never appear in a \NE{} (nor in a best-response set). Therefore, we can think of a game transformation procedure that iteratively detects strictly dominated strategies and sets their payoffs 
to a large negative number. This transformation universally preserves \NEs{}, but it does not universally preserve best-response sets. Note that this game transformation is not separable because its maps $h_{\ks}^i$ now need to take all utility payoffs of the game into consideration, and not only what utility player $i$ receives from strategy profile~$\ks$.

In a similar fashion, one may think of best-response-preserving transformations that are not PATs. This was studied extensively by \citet{Liu96}, who discusses the following example of $3 \times 2$ payoff matrices of \textnormal{PL1} in $2$-player games:
\begin{align}
\label{matrix ex equal BR sets but not PAT eqvl}
A = \begin{pmatrix}
6  & 0\\
0 & 6\\
4 & 4
\end{pmatrix}
\textnormal{ and } \,
A' = \begin{pmatrix}
6  & 0\\
2 & 5\\
4 & 4
\end{pmatrix} \, . 
\end{align}
\begin{figure}[t]
    \centering
    \includegraphics[width=7cm]{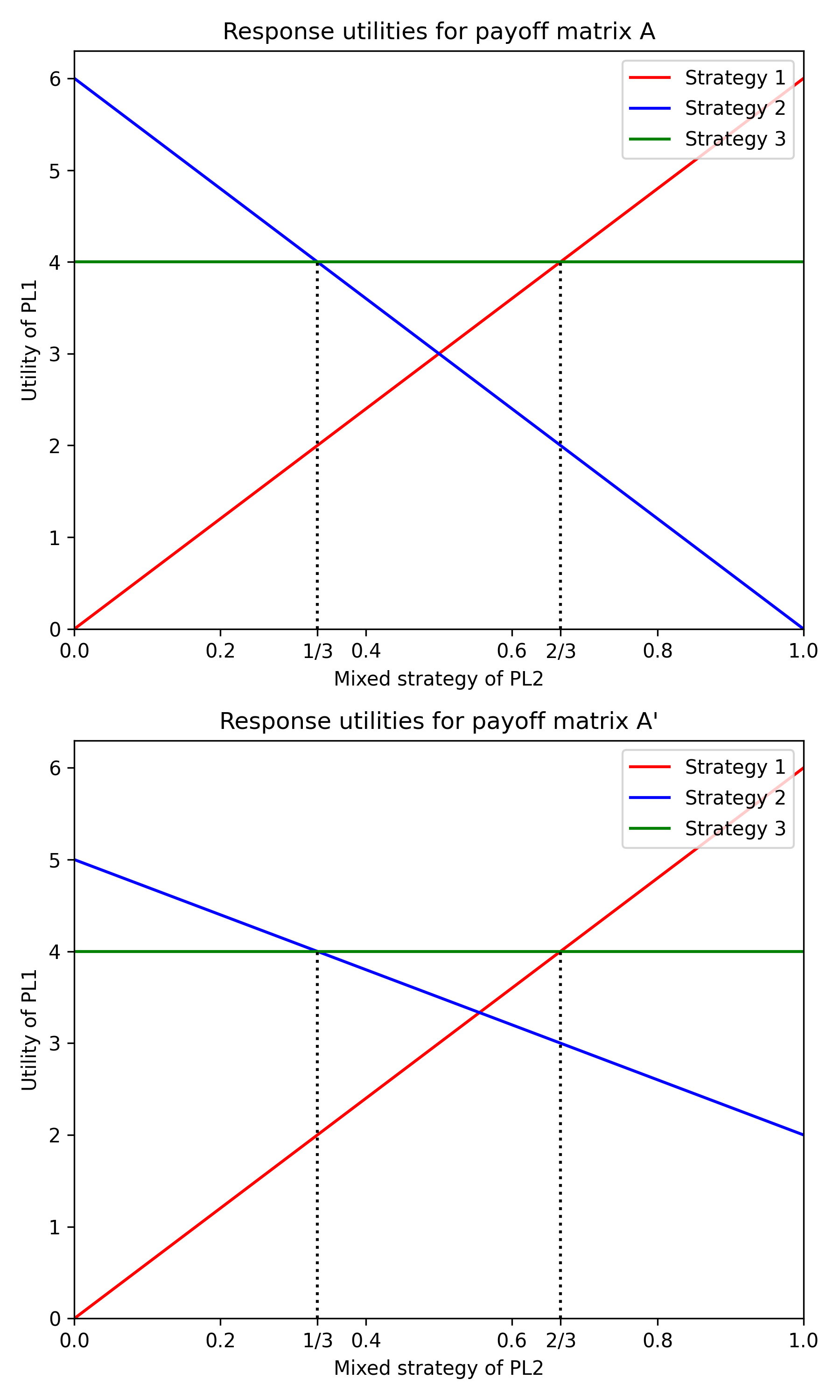}
    \caption{The utility payoffs of each pure strategy $1,2,3$ of \textnormal{PL1} in response to the mixed strategy of \textnormal{PL2} that plays $1$ with probability $x$. Plots correspond to 
    matrices $A$ and $A'$ from (\ref{matrix ex equal BR sets but not PAT eqvl}). The best-response set to a strategy $(x,1-x)$ 
    of \textnormal{PL2} will be all convex combinations of pure strategies of \textnormal{PL1} that are maximal at $x$ in the respective plot.}
    \label{fig:ex equal BR sets but not PAT eqvl}
\end{figure} 
As visualized by Figure~\ref{fig:ex equal BR sets but not PAT eqvl}, the best responses of \textnormal{PL1} to any mixed strategy of \textnormal{PL2} are the same in $A$ and $A'$. However, $A'$ cannot be obtained from $A$ through a PAT: If there were such a PAT, then the payoff from profile $(2,1)$ requires a shift of $c_1^1 = 2$. Hence, the payoff from profile $(1,1)$ requires a scaling of $\alpha^1 = \frac{2}{3}$. But these components of a positive affine transformation do not work out for the payoff from profile $(3,1)$, leaving us with a contradiction.

\citet{Liu96} develops a polynomial-time method, called \textit{bi-affine} transformation, that determines whether two $2$-player normal-form games have the same best-response sets. The procedure detects which strategies and strategy pairs are essential, and derives that only the essential pairs need to be in a positive affine relationship. Hence, the method includes PATs, but it is also more powerful than that. \citeauthor{Liu96}'s dissertation (\citeyear{Liu95old}) extends those ideas to $N$-player games ($N \geq 3$). But in $N$-player games, the method downgrades to a sufficient condition: Two $N$-player games ($N \geq 3$) may have the same best-response sets while not being a \textit{quasi-affine} transformation of each other. Furthermore, their method becomes computationally inefficient. 
In fact, we have shown in Section~\ref{sec:decaboutbrs} 
more generally that determining whether two $3$-players games have the same BR sets is $\mathbf{co\text{-}NP}$-hard. 

\citeauthor{Liu96} concludes with an immediate open problem for future work: to characterize games with the same \NEs{}. To that end, \citet{10.1016/j.tcs.2008.07.021} proves that it is $\mathbf{NP}$-complete to decide whether two $2$-player games share a common \NE{}, and that it is $\mathbf{co\text{-}NP}$-hard to decide whether two $2$-player games have the same \NE{} set. 

In light of these negative results about characterizing best-response equivalence and \NE{} equivalence in full generality - assuming the well-accepted complexity belief $\mathbf{co\text{-}NP} \neq \mathbf{P}$ - we restrict our focus to a subclass of equivalence-preserving transformations based on separability. We argue that among naturally defined subclasses, separable game transformations constitute a maximal subclass for which it is still open whether it contains tractable and equivalence-preserving transformations aside from PATs.

\section{Transformations that Preserve Nash Equilibrium Sets or Best-Response Sets}
\label{sec:strat equiv preserving}
To our knowledge, the results of this section are all novel unless explicitly stated otherwise. They can be summarized in the following statement.
\begin{thm}
\label{equiv charact of PAT}
Let $H = \{H^i\}_{i \in [N]}$ be a separable game transformation. Then:
\begin{align}
\label{mt1} \tag{i}
&\, H \text{ universally preserves \NE{} sets} \\
\label{mt2} \tag{ii}
&\iff \text{for each player~$i$, map $H^i$ universally} 
\\
\nonumber
&\, \quad \quad \quad \text{preserves best responses} \\
\label{mt3} \tag{iii}
&\iff H \text{ is a positive affine transformation}.
\end{align}
\end{thm}

Lemma~\ref{multiplayer PAT preserves} gives (\ref{mt3})$\implies$(\ref{mt1}), and so the novel part of Theorem~\ref{equiv charact of PAT} is the implication chain (\ref{mt1})$\implies$(\ref{mt2})$\implies$(\ref{mt3}). The key property that enables us to develop this chain is that we require the separable game transformations $H = \{H^i\}_{i \in [N]}$ to be \textit{universally} applicable, no matter the game $G = \{u_i\}_{i \in [N]}$ we have at hand. 

We shall state two algorithmic consequences of Theorem~\ref{equiv charact of PAT}.
\begin{cor}
Given two normal-form games, we can decide within polynomial time whether one is a transform of the other through an equivalence-preserving separable game transformation.
\end{cor}
This is because deciding whether a game is a PAT transform of another reduces to solving a linear (in-)equation system for the variables $\{ \alpha^i, C^i \}_{ i \in [N] }$. A case distinction is needed for solution points that take on values $\alpha^i = 0$.

\begin{cor}
Given a 2-player normal-form game $G$, we can find a transform $G'$ of it (if it exists) via an equivalence-preserving separable game transformation, such that $G'$ is a zero-sum or rank-1 game. With that, a Nash equilibrium for $G$ can be computed subsequently. Both take polynomial time.
\end{cor}

This follows from the results in \cite{HeymanG23,main}.

Before tackling Theorem~\ref{equiv charact of PAT}, let us characterize a special property that a game transformation can satisfy in which the strategy choice of player~$i$ does not influence the map that is being used to transform her utilities.
\begin{defn}
\label{defn opponent dependence}
Let map $H^i$ come from a separable game transformation $H$. Then we say that $H^i$ only depends on the strategy choices of the opponents if for all pure strategy choices $\ks_{-i} \in S^{-i}$ of the opponents, we have the map identities 
$h_{1, \ks_{-i}}^i = \ldots = h_{m_i, \ks_{-i}}^i\,: \R \to \R$. 
\end{defn}

Next, we can show (\ref{mt1})$\implies$(\ref{mt2}).
\begin{prop}
\label{two conclusions of strat equiv preserving}
Let $H = \{H^i\}_{i \in [N]}$ be a separable game transformation that universally preserves \NE{} sets and consider the map $H^i$ of a player~$i$. Then $H^i$ only depends on the strategy choices of the opponents. Furthermore, $H^i$ universally preserves best responses.
\end{prop}
\begin{proof}[Proof sketch] \, 
\paragraph{1.} Such a universally preserving transformation $H$ should in particular not change the \NE{} set for a trivial game in which all players receive the same constant utility $z \in \R$ from all strategy profiles. In such a game, the whole strategy set~$S$ will make the \NE{} set. For that to also be the case in the transformed game, we show for every player $i$, that the transformations maps $h_{1, \ks_{-i}}^i, \ldots, h_{m_i, \ks_{-i}}^i$ must all evaluate the same on any input value $z$.
\paragraph{2.} Let $u_i$ be an arbitrary utility function of player~$i$. Complete $u_i$ to a full game $G$ by setting the utilities of all other players to the constant payoff of $0$. This makes any strategy $s^j$ of another player $j \neq i$ always a best-response strategy in $G$. We can then show that this must also hold in the transformed game $H(G)$, using the first conclusion. Therefore, we get the following equivalence chain:
\begin{enumerate}
    \item[(a)] a strategy $s^i$ of player $i$ is a best response to a profile $s^{-i}$ of the opponent players and with respect to $u_i$ if and only if
    \item[(b)] $(s^i, s^{-i})$ is a \NE{} of $G$  if and only if
    \item[(c)] $(s^i, s^{-i})$ is a \NE{} of $H(G)$  if and only if
    \item[(d)] $s^i$ a best response to $s^{-i}$ with respect to $H^i(u_i)$.
\end{enumerate}
\end{proof}
The first conclusion captures the intuition that if the maps $h_{\ks}^i$ from $H^i$ would depend on the strategy choice of player $i$, then in the transformed game $H(G)$, player $i$ may need to adjust her strategy choice to those $h_{\ks}^i$ that map payoffs to high values. This would affect the strategic decision making of player $i$ and therefore the \NE{} set. Similar reasoning provides us with a related (but independent) result:

\begin{lemma}
\label{br implies opp dependent}
If map $H^i$ universally preserves best responses, then $H^i$ only depends on the strategy choices of the opponents.
\end{lemma}

Due to Proposition~\ref{two conclusions of strat equiv preserving}, we can transition to the analysis of transformation maps $H^i$ that universally preserve best responses. Thus from now on, our results also become relevant to game theory research that focuses on best-response sets, such as best-response dynamics  or fictitious play.

Proposition~\ref{two conclusions of strat equiv preserving} moreover allows us to restrict our analysis to the map $H^1$ for \textnormal{PL1} w.l.o.g.~because any results for $H^1$ will analogously also hold for maps $H^2,\ldots,H^N$. By Lemma~\ref{br implies opp dependent}, we can also drop the dependence of $H^1$ on $k_1$ and write 
$H^1 := \big\{ h_{\ks_{-1}}^1 : \R \rightarrow \R \big\}_{\ks_{-1} \in S^{-1}}$.

For each pure-strategy map $h_{\ks_{-1}}^1$ we introduce its \textit{distance distortion} function which takes two utility values and measures their distance after a $h_{\ks_{-1}}^1$-transformation:
\begin{align}
\label{def dist dist fct}
\begin{aligned}
\Delta h_{\ks_{-1}}^1 : \R \times \R \rightarrow \R \, , \, \,
(z,w) \mapsto h_{\ks_{-1}}^1(z) \, - \, h_{\ks_{-1}}^1(w)
\end{aligned}
\end{align}

The following lemma reveals an important preliminary observation on how the distance distortion functions $\Delta h_{\ks_{-1}}^1$ relate to each other. It highlights how the distorted utility distances are connected upon a strategy change of a player $j \neq 1$ from, e.g., some pure strategy $k_j \neq 1$ to their first pure strategy $1 \in [m_j]$. It is again crucial that $H^1$ preserves best responses \textit{universally} in order to deduce these global properties of and connections between the maps within~$H^1$.

\begin{lemma}
\label{relating distance maps}
Suppose transformation map $H^1$ universally preserves best responses. Take a player $j \in [N] \setminus \{1\}$ and profile $\ks_{-1} \in S^{-1}$ with $k_j \neq 1$. Define $\ks_{-1}' \in S^{-1}$ to be the same as $\ks_{-1}$ except for player~$j$'s choice which shall be set to $k_j' = 1$. Then, for all $z,z',w,w' \in \R$:
\\
    $z-w \geq z'-w' \, \, \iff \, \, \Delta h_{\ks_{-1}}^1(z,w) \geq \Delta h_{\ks_{-1}'}^1(z',w')$.
\end{lemma}
\begin{proof}[Proof sketch] Construct a utility function $u^1$ for each set of values for $j, \ks_{-1}, z, z', w$, and $w'$. Namely, set $u_1(1, \ks_{-1}) := z$ and $u_1(1, \ks_{-1}') := w'$, and for all strategies $l \in [m_1] \setminus \{1\}$ set $u_1(l, \ks_{-1}) := w$ and $u_1(l, \ks_{-1}') := z'$. Observe that uniformly randomizing over $\ks_{-1}$ and $\ks_{-1}'$ not only makes a correlated strategy of the opponents, but also a valid mixed strategy profile. Hence, the left hand side of the equivalence 
can be reinterpreted as strategy $1 \in [m_1]$ performing better for player $1$ than any other of her strategies $l \in [m_1] \setminus \{1\}$ if player~$j$ uniformly mixes over strategies $k_j$ and $k_j'$ and if all other players $r \notin \{1,j\}$ play their respective strategy $k_r \in [m_r]$. We then derive the equivalence 
by using that $H^1$ preserves strategy $1$ being such a best response and by using Lemma~\ref{br implies opp dependent}.
\end{proof}

Next, observe that by definition, these distance distortion functions are skew-symmetric, that is, $\forall \, z,w \in \R:$
    \\
    $\Delta h_{\ks_{-1}}^1(z,w) = \, -\Delta h_{\ks_{-1}}^1(w,z)$.

With the following lemma, we further tighten the connection between the pure-strategy maps $h_{\ks_{-1}}^1$ through their distance distortion functions. Last but not least, we shine some light on how those maps $h_{\ks_{-1}}^1$ behave individually in the subsequent lemma. 

\begin{lemma}
\label{same distance distortion functions}
If map $H^i$ universally preserves best responses, then the pure-strategy maps in $H^1$ equally distort distances. That is, 
$\forall \, \ks_{-1} \in S^{-1}:$ $\Delta h_{\ks_{-1}}^1 = \Delta h_{\1_{-1}}^1$, 
where $\1_{-1} := (1, \ldots, 1) \in S^{-1}$.
\end{lemma}
\begin{proof}[Proof sketch] Make iterative use of Lemma~\ref{relating distance maps} for all other players $j \neq 1$, and make use of the skew-symmetry.
\end{proof}

\begin{lemma}
\label{strat specific maps behaviour}
If map $H^i$ universally preserves best responses, then for all $\ks_{-1} \in S^{-1}$:
\begin{enumerate}
\item map $h_{\ks_{-1}}^1$ is strictly increasing, and
\item map $h_{\ks_{-1}}^1$ distorts distances independently of their reference points:
\begin{align*}
\forall z,z',\lambda \in \R \, : \, \Delta h_{\ks_{-1}}^1(z + \lambda, z) = \Delta h_{\ks_{-1}}^1(z' + \lambda, z') \, .
\end{align*}
\end{enumerate}
\end{lemma}
\begin{proof}[Proof sketch] For the first conclusion make use of Lemma~\ref{relating distance maps} for values $z' = w'$, and of Lemma~\ref{same distance distortion functions}. For the second conclusion, utilize skew-symmetry together with the same two lemmata.
\end{proof}

With Lemmata~\ref{same distance distortion functions} and~\ref{strat specific maps behaviour}, we can finally show that positive affine transformations are the only game transformations that universally preserve best responses. Intuitively speaking, the second conclusion of Lemma~\ref{strat specific maps behaviour} states that taking a step of length $\lambda$ in the domain space consistently maps to taking a step of some other length in the range space, independently of the base point $z$ from which we take such a step. This brings us to two known results from the analysis literature. Recall that a function $h: \R \rightarrow \R$ is called linear if there exists some $a \in \R$ such that 
    $\forall z \in \R \, : \, h(z) = az$. 
A function $h: \R \rightarrow \R$ is said to be additive if it satisfies 
    $\forall x,y \in \R \, : \, h(x+y) = h(x) + h(y)$.
\begin{lemma}[\cite{BSMA_1875__9__281_1,Reem2017RemarksOT}]
\label{helping analysis lemma}
If a map $h: \R \rightarrow \R$ is monotone and additive, then it is also linear.
\end{lemma}
\begin{cor}
\label{affine linearity corollary}
Let $h: \R \rightarrow \R$ be monotone and satisfy for all $z,z',\lambda \in \R:$ 
    $h(z+\lambda) - h(z) = h(z' + \lambda) - h(z')$.
\\
Then h is affine linear, i.e., there exist some $a,c \in \R$ such that for all 
    $\forall z \in \R \, : \, h(z) = az + c$.
\end{cor}

This brings us to the completion of this section.
\begin{proof}[Proof sketch of Theorem~\ref{equiv charact of PAT}] \, \\
    Implication~(\ref{mt3})$\implies$(\ref{mt1}) follows from Lemma~\ref{multiplayer PAT preserves}, and implication~(\ref{mt1})$\implies$(\ref{mt2}) follows from Proposition~\ref{two conclusions of strat equiv preserving}. For (\ref{mt2})$\implies$(\ref{mt3}), recall that by symmetry, our results for $H^1$ hold analogously for all maps~$H^i$. By Lemmata~\ref{br implies opp dependent} and~\ref{strat specific maps behaviour}, the maps $h_{\ks}^i = h_{\ks_{-i}}^i$ satisfy the conditions of Corollary~\ref{affine linearity corollary}. Thus, there exist parameters $a_{\ks_{-i}}^i, c_{\ks_{-i}}^i \in \R$ for each~$\ks_{-i} \in S^{-i}$ such that
    $\forall z \in \R \, : \, h_{\ks_{-i}}^i(z) = a_{\ks_{-i}}^i \cdot z + c_{\ks_{-i}}^i$ .
    \\
    Lemma~\ref{same distance distortion functions} implies $a_{\ks_{-i}}^i = a_{\1_{-i}}^i$ for all $\ks_{-i} \in S^{-i}$. Therefore, we only have to keep track of one scaling parameter $\alpha^i$ for all the maps within $H^i$. With the first conclusion of Lemma~\ref{strat specific maps behaviour}, we obtain $\alpha^i > 0$. Putting everything together, we have shown that $H = (H^1, \ldots, H^N)$ makes a positive affine transformation.
    
\end{proof}

Theorem~\ref{equiv charact of PAT} gives two novel equivalent characterizations of PATs that highlight their special status among game transformations: PATs are the only separable game transformations that always preserve the \NE{} set or, respectively, the best-response sets. 

One way to circumvent this result is to focus on game transformations that we only care to apply on particular subclasses of $N$-player games. Preferably, the game properties defining such a subclass would be generic enough to still contain "most" games. On the other hand, one may instead also consider non-separable game transformations as discussed in Section~\ref{sec:discussion}. 

\section{Further Related Literature}
\label{sec:literature review}
Much work has gone into identifying when two games can be considered strategically equivalent.

Strategic similarity, for example, is an important aspect of Potential Games (cf. \citet{MONDERER1996124}). \citet{MORRIS2004260} noted that a game $G$ is a weighted potential game if and only if it is the PAT transformation of an identical interest game\footnote{Identical interest game: Given an action profile $s$, each player shall receive the same utility from $s$.}. They also characterized when two given games are \textit{best-response equivalent}, \textit{better-response equivalent} or \textit{von Neumann-Morgenstern equivalent}. The former and latter are directly tied to our concepts of preserving best-response sets and to PATs. Unfortunately, we were not able to base the second part of Theorem~\ref{equiv charact of PAT} on the insights from \citeauthor{MORRIS2004260} because their characterization for best-response equivalence only holds for games that satisfy specific properties.

\citet{hammond} described that the strategic decision-making in a game in mixed strategies does not depend on the player's numerical utility values, but solely on the preferences that the utility functions induce over the strategies. 
In Appendix~\ref{sec:utility_theory}, we give some further background on utility theory in order to put Hammond's work into our context. Using the Expected Utility Theorem -- cf., e.g., \citet{MasCol} -- \citeauthor{hammond} deduced that utility functions that induce the same preferences can only differ up to a positive affine transformation. 
Note 
that the property of preserving the player's preferences is, in general, strictly harder to satisfy than preserving best responses (and, hence, Nash equilibria). Thus, our Theorem \ref{equiv charact of PAT} generalizes their result to the broader question of strategic equivalence.

Moving to more broadly related work, \citeauthor{GABARRO20116675} \shortcite{GABARRO20116675,gabarro_garcia_serna_2013} gave several complexity-theoretic results for the problem of deciding whether two pure strategy games are \textit{isomorphic} w.r.t.\ a notion of game transformation that can help us understand the symmetries within a game \cite[Chapter 3]{harsanyi1988general}. \citet{McKinsey} and \citet{CoopGamesIsomorphism} studied two notions of game equivalency specific to cooperative games.

Finally, there are other lines of related research that work more explicitly with different notions of transforming a game and preserving strategic features \cite{RM-759-PR,10.2307/1912320,ELMES19941,CASAJUS2003267,BergeVaismanEq,9585431}.

\section{Conclusion}
\label{sec:conclusion}
In this paper, we first gave hardness results about deciding whether a strategy constitutes a best response or whether two games have the same best-response sets. Next, we introduced separable game transformations for multiplayer games, and define the properties (i) universally preserving \NE{} sets and (ii) universally preserving best responses. It is well-known that PATs universally preserve \NE{} sets. We showed that separable game transformations which universally preserve \NE{} sets also universally preserve  best responses. In the subsequent results, we derived further that if a separable game transformation universally preserves best responses then it is a positive affine transformation. 

When faced with a strategic interaction 
it can be highly beneficial to consider equivalent variations of it that are easier to analyze. In this paper, we shed light on why PATs have become the go-to transformation method for that purpose, reinforcing their standing as the standard off-the-shelf approach. 
The current literatures on game theory and on decision making in AI are lacking methods to detect or generate strategically equivalent games, and 
we hope that our results can serve as guidance to the development of any such 
detection or generation toolkit.

\appendix

\section*{Acknowledgments}

We want to thank Sebastian van Strien for the valuable discussions
related to this project and to thank Bernhard von Stengel whose feedback inspired the development of Sections 4 and 6. 
We are also grateful to Vojtěch
Kovařík, Rinor Cakaj, and the anonymous reviewers for their valuable improvement suggestions for this paper. Finally, we thank the
Cooperative AI Foundation, Polaris Ventures (formerly the Center
for Emerging Risk Research) and Jaan Tallinn’s donor-advised fund
at Founders Pledge for financial support.

\bibliographystyle{named}
\bibliography{ijcai24}

\clearpage

\section{Proofs for Section~\ref{sec:decaboutbrs}}
\label{sec:dec proofs}
To our knowledge, the results and proofs of this section are all novel unless explicitly stated otherwise.

\begin{defn*}[{\sc CheckIfEverBR}]
    Given a $3$-player normal-form game, does there exist mixed strategies $\rs \in \Delta(S^2)$ of \textnormal{PL2} and $\strats \in \Delta(S^3)$ of \textnormal{PL3} such that pure strategy $1$ of \textnormal{PL1} is a best response to $(\rs, \strats)$?
\end{defn*}

In general, for computational problems involving $N$-player games $G$ with strategy sets $(S^i)_{i \in [N]}$ and utility functions $(u_i)_{i \in [N]}$, we are interested in their computational complexities in terms of $\sum_i |S^i|$ and the binary encoding of all utility payoffs $\Big( u_i(\strats) \Big)_{\strats \in S, i \in [N]}$. For that, we require that utility payoffs take on rational values only. 

We aim to prove Proposition~\ref{checkifeverbr}:
\begin{prop*}
    {\sc CheckIfEverBR} is $\mathbf{NP}$-hard.
\end{prop*}

We achieve this by a reduction from the {\sc BiClique} problem. Recall that a bipartite graph $G = (V \cup W,E)$ is an undirected graph such that each edge $e \in E$ has one endpoint in $V$ and the other in $W$.
\begin{defn*}[{\sc BiClique}]
    Given a bipartite graph $G = (V \cup W,E)$ and integer $1 \leq K \leq m + n$, are there subsets $V^* \subseteq V$ and $W^* \subseteq W$ with $|V^*| = K = |W^*|$ and $(v,w) \in E$ for all $v \in V^*, w \in W^*$?
\end{defn*}
For problems involving bipartite graphs $G = (V \cup W,E)$, we are interested in their computational complexities in terms of $m := |V|$, $n := |W|$ and $l := |E|$. 

The complexity of {\sc BiClique} is known in the literature.
\begin{lemma*}[\citet{Garey90}, Problem GT24] \,
    {\sc BiClique} is $\mathbf{NP}$-complete.
\end{lemma*}

Before we get to the proof of Proposition~\ref{checkifeverbr}, let us give a trivial yes instance and a trivial no instance of {\sc CheckIfEverBR}. This will be used in the proof.

The trivial yes instance shall be $S^1 = \{1\} = S^2 = S^3$ and $u_1(1,1,1) = 0 = u_2(1,1,1) = u_3(1,1,1)$. Then pure strategy $1$ of \textnormal{PL1} is a best response to $(1,1)$.

The trivial no instance shall be $S^1 = \{1, 2\}$, $S^2 = \{1\} =S^3$, $u_1(1,1,1) = 0 = u_2(\cdot,1,1) = u_3(\cdot,1,1)$ and $u_1(2,1,1) = 1$. Then pure strategy $1$ is strictly dominated by $2$ and therefore never a best response to a profile of the opponents.
\begin{proof}[Proof of Proposition~\ref{checkifeverbr}]
    \, \\
    Reduction from {\sc BiClique}. Let $G = (V \cup W,E)$ and $1 \leq K \leq m + n$ be the {\sc BiClique} instance. 

    \medskip
    \textit{Trivial cases:} Let us first remove a couple of edge cases. They are not mutually exclusive, but one can just check these case conditions in the following order until one is satisfied (if at all).

    Case 1: If $K \geq \min\{m,n\} + 1$. Then we have a no instance of {\sc BiClique}. So construct the trivial no instance of {\sc CheckIfEverBR}.

    Case 2: If $K = m$, check in $\mathcal{O}(nm)$ time by going through $W$ if there are at least $K$-many vertices $w \in W$ that satisfy $(v,w) \in E$ for all $v \in V$. If so, then we have a yes instance of {\sc BiClique} by setting $V^* = V$ and $W^*$ equal to the $K$-many found $w$'s. So construct the trivial yes instance of {\sc CheckIfEverBR}. If they do not exist, however, then we have a no instance of {\sc BiClique} because we couldn't find set $W^*$ of size $K$ that matches the only possibility $V^* = V$. So construct the trivial no instance of {\sc CheckIfEverBR}.

    Case 3: If $K = n$, do the analogous procedure as in Case 2, except with reversed roles for $v$ and $w$.
    
    Case 4: If $K = m - 1$, check in $\mathcal{O}(mnm)$ time if there exists $\bar{v} \in V$ such that are at least $K$-many vertices $w \in W$ that satisfy $(v,w) \in E$ for all $v \in V \setminus \{\bar{v}\}$. If so, then we have a yes instance of {\sc BiClique} by setting $V^* = V \setminus \{\bar{v}\}$ and $W^*$ equal to the $K$-many found $w$'s. So construct the trivial yes instance of {\sc CheckIfEverBR}. If they do not exist, however, then we have a no instance of {\sc BiClique} because we couldn't find set $W^*$ of size $K$ that matches the only possibilities $V^* = V \setminus \{\bar{v}\}$ for some $\{\bar{v}\} \in V$. So construct the trivial no instance of {\sc CheckIfEverBR}.

    Case 5: If $K = n - 1$, do the analogous procedure as in Case 4, except with reversed roles for $v$ and $w$.

    Case 6: If neither of the previous case conditions are satisfied. The rest of this proof is considering this case now.
    
    \medskip
    \textit{Construction of the corresponding {\sc CheckIfEverBR} instance:}
    Set $S^2 = V$, $S^3 = W$ and $S^1 = \{1\} \cup \{v\}_{v \in V} \cup \{w\}_{w \in W} \cup \{(v,w)\}_{(v,w) \notin E}$. Intuitively, we want to interpret a mixed strategy $\rs$ of PL2 as PL2 choosing the support $\supp(\rs) := \{ v \in V : \rs(v) > 0 \}$ as the subset $V^*$ of $V$ for the biclique. Analogously, the support of $\strats$ of PL3 shall give the subset $W^*$ of $W$ for the biclique. We make a strategy $v$ (resp. $w$) of PL1 very attractive for PL1 in the case that PL2 (resp. PL3) play their corresponding strategy $v$ (resp. $w$) with too much probability. This accomplishes that in any potential certificate $(\rs, \strats)$, PL2 and PL3 mix over at least $K$ strategies. We also make a strategy $(v,w) \notin E$ of PL1 very attractive for PL1 in the case that PL2 and PL3 both play their corresponding strategies $v$ and $w$ with any significant probability. This accomplishes that in any potential certificate $(\rs, \strats)$, PL2 and PL3 put non-negligible weight on strategies $v$ and $w$ only if $(v,w) \in E$. Let us proceed with the actual utility payoffs.

    Set $u_2(\cdot,\cdot,\cdot) = 0 = u_3(\cdot,\cdot,\cdot)$ because those payoffs are irrelevant. Next, set $u_1(1,\cdot,\cdot) = 1$. Finally, set 
    \begin{flalign*}
        &\forall v,v' \in V \, : \quad u_1(v,v',\cdot) = 
        \begin{cases}
        K + 1 &\text{if } v = v'\\
        0 &\text{if } v \neq v'
        \end{cases} \, ,&&
    \end{flalign*}
     
    \begin{flalign*}
        &\forall w,w' \in W \, : \quad u_1(w,\cdot,w') = 
        \begin{cases}
        K + 1 &\text{if } w = w'\\
        0 &\text{if } w \neq w'
        \end{cases} \, ,&&
    \end{flalign*}
    and
    \begin{flalign*}
        &\forall (v,w) \notin E \, : \quad u_1 \Big( (v,w),v',w' \Big) =&& \\
        &\begin{cases}
        (m-K) (n-K) (K+1)^2 &\text{if } (v,w) = (v',w')\\
        0 &\text{if } (v,w) \neq (v',w')
        \end{cases} \, .&&
    \end{flalign*}    
    Note that by assumption of not being in Cases 1, 2, and 3, we have $(m-K) (n-K) (K+1)^2 > 0$.

    \medskip
    \textit{Analysis of the corresponding {\sc CheckIfEverBR} instance:}
    First, we observe that for any mixed strategies $\rs \in \Delta(S^2)$ of PL2 and $\strats \in \Delta(S^3)$ of PL3, we have:
    \begin{align*}
        \forall v \in V \, : \, u_1(v,\rs,\strats) &= \sum_{v' \in V, w' \in W} \rs(v') \strats(w') u_1(v,v',w') 
        \\
        &= \sum_{w' \in W} \rs(v) \strats(w') (K + 1) 
        \\
        &= \rs(v) (K + 1) \, ,
    \end{align*}
    and, analogously,
    \begin{align*}
        \forall w \in W \, : \, u_1(w,\rs,\strats) = \strats(w) (K + 1) \, ,
    \end{align*}
    and
    \begin{align*}
        &\forall (v,w) \notin E \, :
        \\
        &u_1 \Big( (v,w),\rs,\strats \Big) = \sum_{v' \in V, w' \ in W} \rs(v') \strats(w') u_1 \Big( (v,w),v',w' \Big) 
        \\
        &= \rs(v) \strats(w) (m-K) (n-K) (K+1)^2 \, .
    \end{align*}
    Therefore, pure strategy $1$ is a best response to $(\rs, \strats)$ if and only if
    \begin{flalign}
    \label{cond1}
    \begin{aligned}
        \forall v \in V \, : \, \rs(v) &= \frac{1}{K + 1} u_1(v,\rs,\strats)&&
        \\
        &\leq \frac{1}{K + 1} u_1(1,\rs,\strats) = \frac{1}{K + 1} \, ,&&
    \end{aligned}
    \end{flalign}
    \begin{flalign}
    \label{cond2}
    \begin{aligned}
        \forall w \in W \, : \, \strats(w) &= \frac{1}{K + 1} u_1(w,\rs,\strats)&& 
        \\
        &\leq \frac{1}{K + 1} u_1(1,\rs,\strats) = \frac{1}{K + 1} \, ,&&
    \end{aligned}
    \end{flalign}    and    
    \begin{flalign}
    \label{cond3}
    \begin{aligned}
        &\forall (v,w) \notin E \, :&&
        \\
        &\rs(v)(m - K)(K + 1) \cdot \strats(w)(n - K)(K + 1)&& 
        \\
        &= u_1 \Big( (v,w),\rs,\strats \Big) \leq u_1(1,\rs,\strats) = 1 \, .&&
    \end{aligned}
    \end{flalign} 
    
    \medskip
    \textit{Equivalence of {\sc BiClique} and its corresponding {\sc CheckIfEverBR} instance:}
    Suppose the {\sc BiClique} instance be yes instance that falls into Case 6. Let furthermore $V^*$ and $W^*$ be a biclique certificate. Then, in the corresponding {\sc CheckIfEverBR} instance, choose the following strategies $(\rs, \strats)$ for PL2 and PL3: for $v \in V$ set 
    \begin{align*}
        \rs(v) = 
        \begin{cases}
        \frac{1}{K + 1} &\text{if } v \in V^*\\
        \frac{1}{m - K} \frac{1}{K + 1} &\text{if } v \notin V^*
        \end{cases} \, ,
    \end{align*}
    and for $w \in W$ set
    \begin{align*}
        \strats(w) = 
        \begin{cases}
        \frac{1}{K + 1} &\text{if } w \in W^*\\
        \frac{1}{n - K} \frac{1}{K + 1} &\text{if } w \notin W^*
        \end{cases} \, .
    \end{align*}
    Vectors $\rs$ and $\strats$ form well-defined mixed strategies because we are not in Cases 1, 2, and 3, and because 
    \begin{align*}
        \sum_{v \in V} \rs(v) &= \sum_{v \in V^*} \rs(v) + \sum_{v \notin V^*} \rs(v) 
        \\
        &= \sum_{v \in V^*} \frac{1}{K + 1} + \sum_{v \notin V^*} \frac{1}{m - K} \frac{1}{K + 1} 
        \\
        &= K\frac{1}{K + 1} + (m-K) \frac{1}{m - K} \frac{1}{K + 1} 
        \\
        &= 1\, ,
    \end{align*}
    and analogously $\sum_{w \in W} \strats(w) = 1$ . Moreover, conditions (\ref{cond1}), (\ref{cond2}), and (\ref{cond3}) are all satisfied. Hence, pure strategy $1$ of PL1 is a best response to $(\rs, \strats)$, and, therefore, the corresponding {\sc CheckIfEverBR} instance a yes instance as well.

    Now suppose the {\sc BiClique} instance falls into Case 6 and the corresponding {\sc CheckIfEverBR} instanc is a yes instance.  Let furthermore $(\rs, \strats)$ be the strategy certificate of PL2 and PL3 to which pure strategy $1$ of PL1 is a best response. Then, in the {\sc BiClique} instance we started with, consider the sets 
    \begin{align}
    \label{barV}
        \bar{V} := \Big\{ v \in V : \rs(v) > \frac{1}{m - K} \frac{1}{K + 1} \Big\} \, ,
    \end{align}
    and
    \[
        \bar{W} := \Big\{ w \in W : \strats(w) > \frac{1}{n - K} \frac{1}{K + 1} \Big\} \, .
    \]
    Then, we have for all $v \in \bar{V}$ and $w \in \bar{W}$:
    \begin{align*}
        \rs(v)(m - K)(K + 1) \cdot \strats(w)(n - K)(K + 1) > 1 \, .
    \end{align*}
    Therefore, since $(\rs, \strats)$ satisfies condition (\ref{cond3}) by assumption, we get for all $v \in \bar{V}$ and $w \in \bar{W}$ that $(v,w) \in E$. Further below, we show $|\bar{V}|, |\bar{W}| \geq K$. Therefore, choose any $V^* \subseteq \bar{V}$ and $W^* \subseteq \bar{W}$ with $|V^*| = K = |W^*|$, and $(V^*, W^*)$ makes a biclique certificate of the {\sc BiClique} instance. This shows that the {\sc BiClique} is therefore a yes instance as well.

    \smallskip
    \textit{Proving the subclaim that $|\bar{V}|, |\bar{W}| \geq K$:} We only prove $|\bar{V}| \geq K$ since $|\bar{W}| \geq K$ is proven analogously. We derive
    \begin{flalign*}
        1  &= \sum_{v \in V} \rs(v) = \sum_{v \in V^*} \rs(v) + \sum_{v \notin V^*} \rs(v)&&
        \\
        &\overset{(\ref{cond1}), (\ref{barV})}{\leq} \sum_{v \in \bar{V}} \frac{1}{K + 1} + \sum_{v \notin \bar{V}} \frac{1}{m - K} \frac{1}{K + 1}&&
        \\
        &= |\bar{V}|\frac{1}{K + 1} + (m - |\bar{V}|)\frac{1}{m - K} \frac{1}{K + 1} \, .&&
    \end{flalign*}
    Recall that $(m - K)(K + 1) > 0$ since we are not in Cases 1 and 2 (resp. 3). Moving all terms to one side in the above inequality chain and multiplying it by $(m - K)(K + 1)$ yields
    \begin{flalign*}
        0  &\leq (m-K)|\bar{V}| + (m - |\bar{V}|) - (m - K)(K + 1)&&
        \\
        &= m|\bar{V}| - K|\bar{V}| + m - |\bar{V}| - mK - m + K^2 + K&&
        \\
        &= m ( |\bar{V}| - K )- K ( |\bar{V}| - K ) - ( |\bar{V}| - K )&&
        \\
        &= (m - K - 1)(|\bar{V}| - K)\, .&&
    \end{flalign*}
    By Cases 1, 2 (resp. 3), and 4 (resp. 5), we have $K \leq m - 2$. Therefore, we can divide the above inequality chain by $m - K - 1$ to obtain $|\bar{V}| \geq K$.
\end{proof}

Proposition~\ref{checkifeverbr} allows us to study the following decision problem next.
\begin{defn*}[{\sc CheckIfSameBRs}]
    Given two $3$-player normal-form games with strategy set $S^1 \times S^2 \times S^3$, do they have the same best-response sets?
\end{defn*}
\begin{thm*}
    {\sc CheckIfSameBRs} is $\mathbf{co\text{-}NP}$-hard.
\end{thm*}
\begin{proof}
    We show that its complement, which we denote as {\sc CheckIfDiffBRs}, is $\mathbf{NP}$-hard by a reduction from {\sc CheckIfEverBR}. Given two $3$-player normal-form games with strategy set $S^1 \times S^2 \times S^3$ and utility functions $(u_i)_i$ and $(u_i')_i$ respectively, {\sc CheckIfDiffBRs} asks whether there exists a mixed strategy profile $\strats^{-i} \in \Delta(S^{-i})$ for some player $i \in \{1,2,3\}$ such that the best-response sets $\BR_{u_i}(\strats^{-i})$ and $\BR_{u_i'}(\strats^{-i})$ differ.
    
    Let $G$ be an instance of {\sc CheckIfEverBR}, that is, a $3$-player normal-form game. Denote its strategy set with $S^1 \times S^2 \times S^3$ and its utility functions with $(u_i)_i$. 
    Determine a strict lower bound $L := -1 + \min_{\strats \in \Delta(S)} \big\{ u_1(\strats) \big\}$ on the utilities PL1 may receive in $G$. Construct another game $G'$ with the same strategy set as $G$ and with utility functions $u_2' := u_2$, $u_3' := u_3$, and
    \begin{align*}
        u_1(k_1,k_2,k_3) := 
        \begin{cases}
        L &\text{if } k_1 = 1\\
        u_1(k_1,k_2,k_3) &\text{if } k_1 \neq 1
        \end{cases}
    \end{align*}
    for all $(k_1,k_2,k_3) \in S$. Then $(G,G')$ shall be the corresponding {\sc CheckIfDiffBRs} instance. Let us prove equivalence.

    Suppose $G$ is a yes instance of {\sc CheckIfEverBR}. Let $(\strats^2, \strats^3) \in \Delta(S^2) \times \Delta(S^3)$ be the strategy certificate to which pure strategy $1$ of PL1 is a best response, i.e., $1 \in \BR_{u_1}(\strats^2, \strats^3)$. PL1 has a second strategy (by the definition of a game), and by construction of $L$, strategy $1$ is strictly dominated by strategy $2$ in $G'$ for PL1. Therefore, $1$ can never be a best response in $G'$ for PL1. In particular, $1 \notin \BR_{u_1}(\strats^2, \strats^3)$. Hence, $(G,G')$ is a yes instance of {\sc CheckIfDiffBRs} as well.

    Suppose $(G,G')$ is a yes instance of {\sc CheckIfDiffBRs}. Since PL2 and PL3 receive the same utilities in $G$ and $G'$, their best-response sets will be equal. Therefore, the difference in best-response sets must be for PL1, that is, there exists a strategy certificate $(\strats^2, \strats^3) \in \Delta(S^2) \times \Delta(S^3)$ for which $\BR_{u_1}(\strats^2, \strats^3) \neq \BR_{u_1'}(\strats^2, \strats^3)$. By Corollary~\ref{convex payoff combi}, this means that the two sets do not contain the same pure best responses. Let us treat three imaginable situations (which are not mutually exclusive) separately.

    Situation 1: We have $1 \in \BR_{u_1}(\strats^2, \strats^3)$ but $1 \notin \BR_{u_1'}(\strats^2, \strats^3)$. Then, we are done because this shows that $G$ is also a yes instance of {\sc CheckIfEverBR}.
    
    Situation 2: There exists a pure strategy $k \in S^1 \setminus \{1\}$ with $k \in \BR_{u_1}(\strats^2, \strats^3)$ but $k \notin \BR_{u_1'}(\strats^2, \strats^3)$. Note that $u_1'$ only differs from $u_1$ in how much utility strategy $1  \in S^1$ yields under $u_1'$ and $u_1$, namely, less under $u_1'$. Thus, if $k \neq 2$ was a maximizer of $u_1( \cdot, \strats^2, \strats^3)$, then it must still be a maximizer of $u_1'( \cdot, \strats^2, \strats^3)$. So this situation will never occur because its premise will never hold.

    Situation 3: There exists a pure strategy $k \in S^1$ with $k \in \BR_{u_1'}(\strats^2, \strats^3)$ but $k \notin \BR_{u_1}(\strats^2, \strats^3)$. Then $k \neq 1$ because $1 \in S^1$ is never a best response in $G'$. Moreover, since the only change from $u_1$ to $u_1'$ is a decreased payoff of strategy $1 \in S^1$, this Situation 3 can only happen if $1 \in S^1$ is the sole maximizer of $u_1( \cdot, \strats^2, \strats^3)$. Thus, we are done because this shows that $G$ is also a yes instance of {\sc CheckIfEverBR}.

    In conclusion, we have shown overall that {\sc CheckIfDiffBRs} is $\mathbf{NP}$-hard, and thus, {\sc CheckIfSameBRs} is $\mathbf{co\text{-}NP}$-hard.
\end{proof}

\section{Proofs for Section~\ref{sec:strat equiv preserving}}
\label{sec:main proofs}

To our knowledge, the results and proofs of this section are all novel.

First, some further notation for this appendix. We may write $u \equiv \lambda$ to refer to a function $u: \mathcal{D} \rightarrow \R$ that is a constant function on its domain $\mathcal{D}$, set to the value $\lambda \in \R$. Moreover, let $e_k \in \R^{n}$ ($M \in \N$) stand for the $k$-th standard basis vector, i.e., with a $1$ in its $k$-th entry and $0$'s anywhere else. Finally, in a game $G$ and for a player $i$ with pure strategy set $S^i = [m_i]$, we identify any pure strategy $k \in [m_i]$ with its corresponding mixed strategy vector in $\Delta(S^i)$ which is exactly the basis vector $e_k \in \R^{m_i}$.

\begin{prop*}
\label{app:two conclusions of strat equiv preserving}
Let $H = \{H^i\}_{i \in [N]}$ be a separable game transformation that universally preserves \NE{} sets and consider the map $H^i$ of a player~$i$. Then $H^i$ only depends on the strategy choices of the opponents. Moreover, $H^i$ universally preserves best responses.
\end{prop*}
\begin{proof} 
Take a separable game transformation $H = \{H^i\}_{i \in [N]}$ that universally preserves \NE{} sets and fix some player~$i$. 
\paragraph{1} Fix a pure strategy choice $\ks_{-i} \in S^{-i}$ of the opponent players and take some arbitrary value $z \in \R$. Consider the game $G = \{u_j\}_{j \in [N]}$ with constant utility functions $u_j \equiv z$ for all $j \in [N]$. Then, the \NE{} set will be the whole strategy space $\Delta(S)$. By assumption on $H$, the transformed game $H(G)$ also has the full strategy space as its set of \NEs{}. In particular, each of the strategy profiles $(1 , \ks_{-i}), \ldots,  (m_i , \ks_{-i})$ will be a \NE{} of the transformed game $H(G)$. Hence, for all $2 \leq l \leq m_i$:

\begin{align*}
h&_{1, k_{-i}}^i(z) \overset{u_i \equiv z}{=} h_{1, \ks_{-i}}^i \Big( u_i(1, \ks_{-i}) \Big) \overset{
\textnormal{Def } \ref{def game trafo}}{=} H^i(u_i)(1, \ks_{-i}) \\
\overset{\text{Nash-Eq}}&{=} \max_{t^i \in \Delta(S^i)} \Big\{ \, H^i(u_i)(t^i, \ks_{-i}) \Big\} \overset{\text{Nash-Eq}}{=} H^i(u_i)(l,\ks_{-i}) \\
&= h_{l, \ks_{-i}}^i \Big( u_i(l, \ks_{-i}) \Big) = h_{l, \ks_{-i}}^i(z) \, .
\end{align*}
Since $z$ and $l$ were chosen arbitrarily, we get 
\[h_{1, \ks_{-i}}^i = \ldots = h_{m_i, \ks_{-i}}^i \, . \]
\paragraph{2} Fix player~$i$'s utility function $u_i$ and the opponents' strategy choices $\strats^{-i} \in \Delta(S^{-i})$. Then by \ref{app:BR equal through PBR}, it suffices to identify the pure strategies in the best-response sets $\BR_{u_i}(\strats^{-i})$ and $\BR_{H^i(u_i)}(\strats^{-i})$.

Complete the prefixed $u_i$ to a full game $G = \{u_j\}_{j \in [N]}$ by setting $u_j \equiv 0$ for the other players $j \neq i$. Then, the best-response set of a player~$j \neq i$ is her whole strategy space $\Delta(S^j)$. By assumption on the game transformation~$H$, we get for a pure strategy $e_l = l \in S^i$:

\begin{align*}
&e_l \in \BR_{u_i}(\strats^{-i}) \\
&\iff (e_l, \strats^{-i}) \text{ is a \NE{} for the game } G \\
&\iff (e_l, \strats^{-i}) \text{ is a \NE{} for the game } H(G) \\
\overset{\text{def}}&{\iff} e_l \in \BR_{H^i(u_i)}(\strats^{-i}) \text{ and } \forall j \neq i \, :\\
&\, \quad \quad \quad s^j \in \BR_{H^j(u_j)}(s^1, \ldots, s^{j-1}, s^{j+1}, \ldots, \\
&\, \quad \quad \quad \quad \quad \quad \quad \quad \quad \quad \quad \quad s^{i-1}, e_l, s^{i+1}, \ldots, s^N) \\
\overset{(*)}&{\iff} e_l \in \BR_{H^i(u_i)}(\strats^{-i})
\end{align*}
Let us give some further explanation for step $(*)$. Recall the definition for a strategy $s^j$, $j \neq i$, to be a best response to the opponents' strategy choices $(s^1, \ldots, s^{j-1}, s^{j+1}, \ldots, s^{i-1}, s^i := e_l, s^{i+1}, \ldots s^N)$:
\begin{align*}
    s^j &\in \argmax_{t^j \in \Delta(S^j)} \Big\{ \sum_{\ks \in S} s_{k_1}^1 \cdot \ldots \cdot s_{k_{i-1}}^{i-1} \\
    &\quad \quad \quad \quad \quad \quad \quad \, \, \cdot t_{k_i}^i \cdot s_{k_{i+1}}^{i+1} \cdot \ldots \cdot s_{k_N}^N \cdot h_{\ks}^j \big( u_j(\ks) \big) \, \Big\} \, .
\end{align*}
We can show that the term in the argmax is constant in $t^j$. First, note that the maps $h_{\ks}^j$ are independent of player $j$'s action, which, in particular, implies $h_{\ks}^j = h_{1, \ks_{-j}}^j$. Then, rearranging yields
\begin{align*}
&\sum_{\ks \in S} s_{k_1}^1 \cdot \ldots \cdot s_{k_{i-1}}^{i-1} \cdot t_{k_i}^i \cdot s_{k_{i+1}}^{i+1} \cdot \ldots \cdot s_{k_N}^N \cdot h_{\ks}^j \big( u_j(\ks) \big) \\
\overset{u_j \equiv 0}&{=} \sum_{\ks_{-j}} \Bigg( s_{k_1}^1 \cdot \ldots \cdot s_{k_{j-1}}^{j-1} \cdot s_{k_{j+1}}^{j+1} \cdot \ldots \cdot s_{k_N}^N \\
    &\quad \quad \quad \quad \quad \quad \quad \quad \quad \quad \quad \quad \quad \cdot h_{1, \ks_{-j}}^j ( 0 ) \cdot \sum_{k_j=1}^{m_j} t_{k_j}^j \, \Bigg) \\
\overset{(\dagger)}&{=} \sum_{\ks_{-j}} s_{k_1}^1 \cdot \ldots \cdot s_{k_{j-1}}^{j-1} \cdot s_{k_{j+1}}^{j+1} \cdot \ldots \cdot s_{k_N}^N \cdot h_{1, \ks_{-j}}^j ( 0 ) \, .
\end{align*}
Since the term in the argmax is constant in $t^j$, any strategy of player~$j$ is a best response to $(s^1, \ldots, s^{j-1}, s^{j+1}, \ldots, s^{i-1}, e_l, s^{i+1}, \ldots s^N)$. Therefore, we obtain the equivalence $(*)$ by removing/adding the redundant condition on each $s^j$, $j \neq i$, to be a best response.

All in all, we proved that the sets $\BR_{u_i}(\strats^{-i})$ and $\BR_{H^i(u_i)}(\strats^{-i})$ contain the same pure strategies. Corollary~\ref{app:BR equal through PBR} therefore yields set equality.
\end{proof}

\begin{lemma*}
\label{app:br implies opp dependent}
Suppose a map $H^i$ universally preserves best responses. Then $H^i$ only depends on the strategy choices of the opponents.
\end{lemma*}
\begin{proof} 
Let the pure strategy choice of the opponents be $\ks_{-i} \in S^{-i}$. Pick some $z \in \R$ and set $u_i \equiv z$. Then we can reformulate
\begin{align*}
&h_{1, \ks_{-i}}^i(z) = \ldots = h_{m_i, \ks_{-i}}^i(z) \\
&\iff \forall l \in [m_i]: \, h_{l, \ks_{-i}}^i(z) = \max_{p \in [m_i]} h_{p, \ks_{-i} }^i(z) \\
\overset{u_i \equiv z}&{\iff} \forall l \in [m_i]: \\
&\quad \quad \quad \quad h_{l, \ks_{-i}}^i(u_i(l, \ks_{-i})) = \max_{p \in [m_i]} h_{p, \ks_{-i}}^i(u_i(p, \ks_{-i})) \\
&\iff \forall l \in [m_i]: \, H^i(u_i)(l, \ks_{-i}) = \max_{p \in [m_i]} H^i(u_i)(p, \ks_{-i}) \\
&\iff \forall l \in [m_i]: \, e_l = l \in \BR_{H^i(u_i)}(s^{-i} = \ks_{-i}) \\
\overset{(*)}&{\iff} \forall l \in [m_i]: \, e_l = l \in \BR_{u_i}(s^{-i} = \ks_{-i}) \\
&\iff \forall l \in [m_i]: \, u_i(l, \ks_{-i}) = \max_{p \in [m_i]} u_i(p \ks_{-i}) \\
\overset{u_i \equiv z}&{\iff} \forall l \in [m_i]: \, z = \max_{p \in [m_i]} z \, .
\end{align*}
In $(*)$, we use that $H^i$ is universally best-response preserving.

With the last line of the equivalence chain above being a universally true statement, we obtain that the first line also holds true. Since $z$ was chosen arbitrarily, we can conclude $h_{1, k_{-i}}^i = \ldots = h_{m_i, k_{-i}}^i$.
\end{proof}

\begin{rem*}
A distance distortion function $\Delta h_{\ks_{-1}}^1$, as defined in (\ref{def dist dist fct}), is skew-symmetric:
\begin{align}
\label{app:skew-symm}
    \forall \, z,w \in \R \, : \quad \Delta h_{\ks_{-1}}^1(z,w) = \, -\Delta h_{\ks_{-1}}^1(w,z) \, .
\end{align}
\end{rem*}

The following lemma reveals an important preliminary observation on how the distance distortion functions $\Delta h_{\ks_{-1}}^1$ relate to each other. It highlights how the distorted utility distances are affected by a strategy change of a player $j \neq 1$ from, e.g., some pure strategy $k_j \in [m_j] \setminus \{1\}$ to their pure strategy $1 \in [m_j]$.

We formulate the lemma with index variables $\ps_{-1} = (p_2, \ldots, p_N)$ instead of $\ks_{-1} = (k_2,\ldots,k_N)$ in order to avoid confusion in the proof of the subsequent Lemma~\ref{same distance distortion functions}.
\begin{lemma}
\label{app:relating distance maps}
Suppose transformation map $H^1$ universally preserves best responses. Take a player $j \in [N] \setminus \{1\}$ and profile $\ps_{-1} \in S^{-1}$ with $p_j \neq 1$. Define $\ps_{-1}' \in S^{-1}$ to be the same as $\ps_{-1}$ except for player~$j$'s choice which shall be set to $p_j' = 1$. Then, for all $z,z',w,w' \in \R$:
\begin{align*}
z-w \geq z'-w' \quad \iff \quad \Delta h_{\ps_{-1}}^1(z,w) \geq \Delta h_{\ps_{-1}'}^1(z',w') \, .
\end{align*}
\end{lemma}

\begin{proof} 
Take a transformation map $H^1$ that universally preserves the best-response sets. Then by Lemma~\ref{br implies opp dependent}, its maps $h_{\ks}^1$ only depend on the strategy choices $\ks_{-1}$ of the opponents. Fix $j, \ps_{-1}, \ps_{-1}'$ and $z,z',w,w'$ as described in the lemma statement. We will construct a utility function $u_1$ for these parameters such that a universally best-response preserving $H^1$ reveals to satisfies the property of this lemma.

Set $u_1(1, \ps_{-1}) := z$ and $u_1(1, \ps_{-1}') := w'$. Additionally, for all pure strategies $l \in [m_1] \setminus \{1\}$, set $u_1(l, \ps_{-1}) := w$ and $u_1(l, \ps_{-1}') := z'$. All these utility value assignments are possible because of $p_j \neq 1 = p_j'$. The utility payoffs of PL1 (i.e., the values of $u_1$) from other pure strategy outcomes $\ks \in S$ can be set arbitrarily. Finally, consider the opponents' mixed strategy profile $\strats^{-1} := \frac{1}{2} \ps_{-1} + \frac{1}{2} \ps_{-1}' \in \Delta(S^{-1})$. 
Then we derive:
\begin{flalign*}
&z-w \geq z'-w'&&
\\
&\iff \forall l \in [m_1] \setminus \{1\} \,: &&
\\
&\quad \quad \quad u_1(1, \ps_{-1}) - u_1(l, \ps_{-1}) \geq u_1(l, \ps_{-1}') - u_1(1, \ps_{-1}')&&
\end{flalign*}
\textit{Reorder and divide by $2$}
\begin{flalign*}
&\iff  \forall l \in [m_1] \setminus \{1\} \,: && \\
&\quad \quad \quad \frac{1}{2}u_1(1, \ps_{-1}) + \frac{1}{2} u_1(1, \ps_{-1}') && \\
&\quad \quad \quad \quad \quad \geq \frac{1}{2} u_1(l, \ps_{-1}) + \frac{1}{2} u_1(l, \ps_{-1}')&&
\\
&\iff  \forall l \in [m_1] \setminus \{1\} \,: \, u_1(e_1, \strats^{-1}) \geq u_1(e_l, \strats^{-1})&&
\\
&\iff  e_1 \in \BR_{u_1}(\strats^{-1})&&
\end{flalign*}
\textit{$H^1$ is universally preserves best responses}
\begin{flalign*}
&\iff e_1 \in \BR_{H^1(u_1)}(\strats^{-1})&&
\\         
&\iff  \forall l \in [m_1] \setminus \{1\} \,: && \\
&\quad \quad \quad  H^1(u_1)(e_1, \strats^{-1}) \geq H^1(u_1)(e_l, \strats^{-1})&&
\\
&\iff \forall l \in [m_1] \setminus \{1\} \,: && \\
&\quad \quad \quad \frac{1}{2} h_{1, \ps_{-1}}^1(u_1(1, \ps_{-1}))  + \frac{1}{2} h_{1, \ps_{-1}'}^1(u_1(1, \ps_{-1}'))&&
\\
&\quad \quad \quad \quad \geq \frac{1}{2} h_{l, \ps_{-1}}^1(u_1(l, \ps_{-1})) + \frac{1}{2} h_{l, \ps_{-1}'}^1(u_1(l, \ps_{-1}'))&&
\\
&\iff \forall l \in [m_1] \setminus \{1\} \,: && \\
&\quad \quad \quad h_{1, \ps_{-1}}^1(z)  + h_{1, \ps_{-1}'}^1(w') \geq h_{l, \ps_{-1}}^1(w) + h_{l, \ps_{-1}'}^1(z')&&
\end{flalign*}
\textit{$H^1$ does not depend on the pure strategy choice of player 1}
\begin{flalign*}
&\iff h_{\ps_{-1}}^1(z) + h_{\ps_{-1}'}^1(w') \geq h_{\ps_{-1}}^1(w) + h_{\ps_{-1}'}^1(z')&&
\\
&\iff h_{\ps_{-1}}^1(z) - h_{\ps_{-1}}^1(w) \geq h_{\ps_{-1}'}^1(z') - h_{\ps_{-1}'}^1(w')&&
\\
&\iff \Delta h_{\ps_{-1}}^1(z,w) \geq \Delta h_{\ps_{-1}'}^1(z',w')&&
\end{flalign*}
\end{proof}

\begin{lemma*}
\label{app:same distance distortion functions}
Suppose transformation $H^1$ universally preserves best responses. Then the pure-strategy maps in $H^1$ equally distort distances:
\[\forall \ks_{-1} \in S^{-1} \, : \, \Delta h_{\ks_{-1}}^1 = \Delta h_{\1_{-1}}^1 \]
where $\1_{-1} := (1, \ldots, 1) \in S^{-1}$.
\end{lemma*}

\begin{proof} 
Take a transformation map $H^1$ that universally preserves the best-response sets. Then by Lemma~\ref{br implies opp dependent}, its maps $h_{\ks}^1$ only depend on the strategy choices $\ks_{-1}$ of the opponents. Fix $\ks_{-1} \in S^{-1}$. Recall that the elements $j \geq 2$ and $\ps \in S^{-1}$ in \ref{app:relating distance maps} can be chosen arbitrarily\footnote{We required $p_j \neq 1$, but this is irrelevant for the argument we are making here.}. So we can apply \ref{app:relating distance maps} on a trivially true statement to get for all $z,w \in \R$:
\begin{align*}
&z-w \geq z-w \\
&\implies \forall j \in [N] \setminus \{1\}: \\
&\quad \quad \quad \Delta h_{k_2, \ldots, k_{j-1}, k_j, 1, \ldots, 1}^1(z,w) \\
&\quad \quad \quad \quad \geq \Delta h_{k_2, \ldots, k_{j-1}, 1, 1, \ldots, 1}^1(z,w) \\
&\implies \Delta h_{k_2, \ldots, k_{N-1}, k_N}^1(z,w) \geq \Delta h_{k_2, \ldots, k_{N-1}, 1}^1(z,w) \\
&\quad \quad \quad \quad \geq \ldots \geq \Delta h_{1, \ldots, 1}^1(z,w) \, . 
\end{align*}
With skew-symmetry, we similarly obtain
\begin{align*}
&w-z \geq w-z \\
&\implies \forall j \in [N] \setminus \{1\}: \\
&\quad \quad \quad \Delta h_{k_2, \ldots, k_{j-1}, k_j, 1, \ldots, 1}^1(w,z) \\
&\quad \quad \quad \quad \geq \Delta h_{k_2, \ldots, k_{j-1}, 1, 1, \ldots, 1}^1(w,z) \\
&\implies \Delta h_{k_2, \ldots, k_{N-1}, k_N}^1(w,z) \geq \Delta h_{k_2, \ldots, k_{N-1}, 1}^1(w,z) \\
&\quad \quad \quad \quad \geq \ldots \geq \Delta h_{1, \ldots, 1}^1(w,z) \\
\overset{\cdot \, (-1)}&{\implies} \Delta h_{k_2, \ldots, k_{N-1}, k_N}^1(z,w) \leq \Delta h_{1, \ldots, 1}^1(z,w) \, . 
\end{align*}
Putting both together, we have for all $z,w \in \R$:
\begin{align*}
\Delta h_{\ks_{-1}}^1(z,w) &= \Delta h_{k_2, \ldots, k_{N-1}, k_N}^1(z,w) 
\\
&= \Delta h_{1, \ldots, 1}^1(z,w) = \Delta h_{\1_{-1}}^1(z,w) \, .
\end{align*}

\end{proof}

\begin{lemma*}
\label{app:strat specific maps behaviour}
Suppose transformation $H^1$ universally preserves best responses. Then we obtain for all $\ks_{-1} \in S^{-1}$ that
\begin{enumerate}
\item map $h_{\ks_{-1}}^1$ is strictly increasing, and that
\item map $h_{\ks_{-1}}^1$ distorts distances independently of their reference points:
\begin{align}
\label{app:distance distortion ind of point}
\forall z,z',\lambda \in \R \, : \, \Delta h_{\ks_{-1}}^1(z + \lambda, z) = \Delta h_{\ks_{-1}}^1(z' + \lambda, z') \, .
\end{align}
\end{enumerate}
\end{lemma*}

\begin{proof}
Take a transformation map $H^1$ that universally preserves the best-response sets. Then by Lemma~\ref{br implies opp dependent}, its maps $h_{\ks}^1$ only depend on the strategy choices $\ks_{-1}$ of the opponents.

\paragraph{1} 
Let us first consider $h_{2,1, \ldots, 1}^1$ that is associated to the pure strategy profile $(2,1, \ldots, 1) \in S^{-1}$. Apply \ref{app:relating distance maps} in the following line $(*)$ with parameters $j = 2$, $\ps_{-1} = (2,1, \ldots, 1)$, and $z' = w' \in \R$ to get for arbitrary $z,w \in \R$:
\begin{align*}
z \geq w &\iff z-w \geq 0 = z' - w' \\
\overset{(*)}&{\iff} \Delta h_{2,1, \ldots, 1}^1(z,w) \geq \Delta h_{\1_{-1}}^1(z',w') \overset{z' = w'}{=} 0 \\
&\iff h_{2,1, \ldots, 1}^1(z) \geq h_{2,1, \ldots, 1}^1(w) \, .
\end{align*}
Consequently, we have for arbitrary $\bar{z},\bar{w} \in \R$: 
\begin{align*}
\bar{z} > \bar{w} &\iff \bar{z} \geq \bar{w} \text{ and } \bar{w} \ngeq \bar{z} \\
\overset{\text{by above}}&{\iff} h_{2,1, \ldots, 1}^1(\bar{z}) \geq h_{2,1, \ldots, 1}^1(\bar{w})\\
&\, \quad \quad \, \, \, \text{ and }  h_{2,1, \ldots, 1}^1(\bar{w}) \ngeq h_{2,1, \ldots, 1}^1(\bar{z}) \\
&\iff h_{2,1, \ldots, 1}^1(\bar{z}) > h_{2,1, \ldots, 1}^1(\bar{w}) \, .
\end{align*}
This shows that $h_{2,1, \ldots, 1}^1$ is strictly increasing.
\\
For arbitrary $\ks_{-1} \in S^{-1}$, we can then use Lemma~\ref{same distance distortion functions} to obtain
\begin{align*}
\bar{z} > \bar{w} &\iff h_{2,1, \ldots, 1}^1(\bar{z}) > h_{2,1, \ldots, 1}^1(\bar{w}) \\
&\iff \Delta h_{2,1, \ldots, 1}^1(\bar{z}, \bar{w}) > 0 \\
&\iff \Delta h_{\ks_{-1}}^1(\bar{z}, \bar{w}) = \Delta h_{\1_{-1}}^1(\bar{z}, \bar{w}) \\
&\quad \quad \quad \, = \Delta h_{2,1, \ldots, 1}^1(\bar{z}, \bar{w}) > 0 \\
&\iff h_{\ks_{-1}}^1(\bar{z}) > h_{\ks_{-1}}^1(\bar{w}) \, .
\end{align*}
Thus, $h_{\ks_{-1}}^1$ is strictly increasing as well. 

\paragraph{2} 
Because of Lemma~\ref{same distance distortion functions}, we only need to show that the map $\Delta h_{\1_{-1}}^1$ satisfies property (\ref{app:distance distortion ind of point}), which would consequently imply the property for all maps $\Delta h_{\ks_{-1}}^1$.

Fix $z,z',\lambda \in \R$. Then the following equivalence chain uses skew-symmetry (\ref{app:skew-symm}) in $(*)$, Lemma~\ref{same distance distortion functions} in $(\dagger)$, and \ref{app:relating distance maps} in $(\star)$ for parameters $j = 2$ and $\ps_{-1} = (2,1, \ldots, 1)$:
\begin{align*}
&\Delta h_{\1_{-1}}^1(z + \lambda, z) = \Delta h_{\1_{-1}}^1(z' + \lambda, z') 
\\
\overset{(*)}&{\iff} \Delta h_{\1_{-1}}^1(z + \lambda, z) \geq \Delta h_{\1_{-1}}^1(z' + \lambda, z') \\
&\, \quad \quad \quad \text{and} \, \, \Delta h_{\1_{-1}}^1(z, z + \lambda)\geq \Delta h_{\1_{-1}}^1(z', z' + \lambda) 
\\
\overset{(\dagger)}&{\iff} \Delta h_{2,\ldots, 1}^1(z + \lambda, z) \geq \Delta h_{\1_{-1}}^1(z' + \lambda, z') \\
&\, \quad \quad \quad \text{and} \, \, \Delta h_{2,\ldots, 1}^1(z, z + \lambda)\geq \Delta h_{\1_{-1}}^1(z', z' + \lambda) 
\\
\overset{(\star)}&{\iff} z + \lambda - z \geq z' + \lambda - z' \\
&\, \quad \quad \, \, \, \, \text{and}  \, \, z - (z + \lambda) \geq z' - (z' + \lambda) \, .
\end{align*}
The last line is a true statement and thus, the first line as well. Because $z,z',\lambda \in \R$ were taken arbitrarily, map $h_{\1_{-1}}^1$ satisfies property (\ref{app:distance distortion ind of point}).
\end{proof}

\begin{thm*}
Let $H = \{H^i\}_{i \in [N]}$ be a separable game transformation. Then:
\begin{align}
\label{app:mt1} \tag{i}
&\, H \text{ universally preserves \NE{} sets} \\
\label{app:mt2} \tag{ii}
&\iff \text{for each player~$i$, map $H^i$ universally} 
\\
\nonumber
&\, \quad \quad \quad \text{preserves best responses} \\
\label{app:mt3} \tag{iii}
&\iff H \text{ is a positive affine transformation}.
\end{align}
\end{thm*}

\begin{proof} 
See main body.
\end{proof}

\section{Helping Lemmas}
\label{sec:helpinglemmas}

This appendix section does not contain original ideas and is just included for completeness.

Denote the restriction of a best-response set to its pure strategies as $\PBR_{u_i}(\strats^{-i}) :=  \BR_{u_i}(\strats^{-i}) \cap \{e_1, \ldots, e_{m_i}\}$. Then, we have that best responses are always convex combinations of pure best responses:
\begin{lemma*}
\label{app:BR charact}
Take a game $G = \{u_i\}_{i \in [N]}$, fix a player~$i \in [N]$ and a strategy profile $\strats^{-i} \in \Delta(S^{-i})$ of the opponents. Then, we have for $t^i \in \Delta(S^i)$:
\begin{align}
\begin{aligned}
    &t^i \in \BR_{u_i}(\strats^{-i}) \\
    &\iff \forall k \in [m_i] : \,  t_k^i = 0 \, \, \, \textup{or} \, \, \, e_k \in \PBR_{u_i}(\strats^{-i})  \, .
\end{aligned}
\end{align}
\end{lemma*}
\begin{proof}
We can observe
\begin{align}
\label{convex payoff combi}
\begin{aligned}
    &u_i(t^i, \strats^{-i}) \\
    &= \sum_{\ks \in S} s_{k_1}^1 \cdot \ldots \cdot s_{k_{i-1}}^{i-1} \cdot t_{k_i}^i \cdot s_{k_{i+1}}^{i+1} \cdot \ldots \cdot s_{k_N}^N \cdot u_i(\ks) 
    \\
    &= \sum_{k_i=1}^{m_i} t_{k_i}^i \cdot \sum_{\ks_{-i} \in S^{-i}} s_{k_1}^1 \cdot \ldots \cdot s_{k_{i-1}}^{i-1} \cdot s_{k_{i+1}}^{i+1} \cdot \ldots \\
    &\quad \quad \quad \quad \quad \quad \quad \quad \quad \quad \quad \quad \quad \quad \quad \quad \quad \cdot s_{k_N}^N \cdot u_i(\ks) 
    \\
    &= \sum_{k_i=1}^{m_i} t_{k_i}^i \cdot u_i(e_{k_i}, \strats^{-i}) = \sum_{k=1}^{m_i} t_{k}^i \cdot u_i(e_k, \strats^{-i}) \, .
\end{aligned}
\end{align}
Thus, the mixed strategy $t^i$ of player $i$ only determines the convex combination of the attainable utility values $\Big( u_i(e_k, \strats^{-i}) \Big)_k$. Therefore, any best-response strategy $t^i$ must only randomize over maximal values within $\Big( u_i(e_k, \strats^{-i}) \Big)_k$, that is, over pure best-response strategies.
\end{proof}

\begin{cor}
\label{app:BR equal through PBR}
Two best-response sets (of possibly different games) are equal if and only if they contain the same pure best responses.
\end{cor}

\begin{lemma*}
\label{app:multiplayer PAT preserves}
Take a PAT $H_{\textnormal{PAT}} = \big\{ \alpha^i, C^i \big\}_{i \in [N]}$ and any game $G = \{u_i\}_{i \in [N]}$. Then, the transformed game $H_{\textnormal{PAT}}(G) = \{u_i'\}_{i \in [N]}$ has the same best-response sets as $G$. Consequently, $H_{\textnormal{PAT}}(G)$ also has the same \NE{} set as $G$.
\end{lemma*}
\begin{proof}
The proof is an appropriate generalization of the known proof for Lemma~\ref{PAT preserves lemma}.

Take a game $\{u_i\}_{i \in [N]}$, fix a player~$i$ and the opponents' strategy choices $\strats^{-i}$. Then, we have
\begin{flalign*}
    &\BR_{u_i'}(\strats^{-i}) = \argmax_{t^i \in \Delta(S^i)} \Big\{ \, u_i'(t^i,\strats^{-i}) \, \Big\} \\
    &= \argmax_{t^i \in \Delta(S^i)} \Big\{ \\
    &\quad \quad \sum_{\ks \in S} s_{k_1}^1 \cdot \ldots \cdot s_{k_{i-1}}^{i-1} \cdot t_{k_i}^i \cdot s_{k_{i+1}}^{i+1} \cdot \ldots \cdot s_{k_N}^N \cdot u_i'(\ks) \, \Big\} \\
    \overset{(\ref{PAT transformed utilities})}&=\argmax_{t^i \in \Delta(S^i)} \Big\{ \sum_{\ks \in S} s_{k_1}^1 \cdot \ldots \cdot s_{k_{i-1}}^{i-1} \cdot t_{k_i}^i \cdot s_{k_{i+1}}^{i+1} \cdot \ldots  \\
    &\quad \quad \quad \quad \quad \quad \quad \quad \quad \quad \quad \quad \cdot s_{k_N}^N \cdot  \big( \alpha^i \cdot u_i(\ks) +  c_{\ks_{-i}}^i \big) \Big\} \\
    \overset{(*)}&{=} \argmax_{t^i \in \Delta(S^i)} \Big\{ \\
    &\quad \quad \alpha^i \cdot \sum_{\ks \in S} s_{k_1}^1 \cdot \ldots \cdot s_{k_{i-1}}^{i-1} \cdot t_{k_i}^i \cdot s_{k_{i+1}}^{i+1} \cdot \ldots \cdot s_{k_N}^N \cdot u_i(\ks) \\
    &\quad \quad + \sum_{\ks_{-i} \in S^{-i}} s_{k_1}^1 \cdot \ldots \cdot s_{k_{i-1}}^{i-1} \cdot s_{k_{i+1}}^{i+1} \cdot \ldots \cdot s_{k_N}^N \cdot c_{\ks_{-i}}^i  \cdot 1\Big\} \\
\end{flalign*}
\begin{flalign*}
    \overset{(\dagger)}&{=} \argmax_{t^i \in \Delta(S^i)} \Big\{ \\
    &\quad \quad \sum_{\ks \in S} s_{k_1}^1 \cdot \ldots \cdot s_{k_{i-1}}^{i-1} \cdot t_{k_i}^i \cdot s_{k_{i+1}}^{i+1} \cdot \ldots \cdot s_{k_N}^N \cdot u_i(\ks) \Big\} \\
    &= \argmax_{t^i \in \Delta(S^i)} \Big\{ \, u_i(t^i,\strats^{-i}) \, \Big\} = \BR_{u_i}(\strats^{-i}) 
\end{flalign*}
We obtain the second summand in $(*)$ by changing the order of summation and multiplication such that $\sum_{k_i = 1}^{m_i} t_i$ remains as the most inner sum. Since $\sum_{k_i = 1}^{m_i} t_i = 1$, this factor can be dropped. We get line $(\dagger)$ because the argmax operator is neither affected by a constant in $t_i$ (such as the secoond summand) nor by rescaling with a positive factor (such as $\alpha_i$).

Finally, the definition of a \NE{} immediately implies that strategy profile $s$ is a \NE{} for the PAT transformed game $\{u_i'\}_{i \in [N]}$ if and only if it was one for the original game $\{u_i\}_{i \in [N]}$.
\end{proof}

\section{Monotone and additive implies linear}
\label{sec:analysislemma}
The proof of the following lemma is taken from \citeauthor{monadd1}~[\citeyear{monadd1,monadd2}] and just included for completeness.
\begin{lemma*}
Take a map $h: \R \rightarrow \R$ which is monotone and additive. Then:
\begin{enumerate}

\item $ h(0) = 0 $ .
\item \label{hihi} $ \forall x \in \R \, : \quad -h(-x) = h(x) $ .
\item $ \forall n \in \N, x \in \R \, : \quad h(n \cdot x) = n \cdot h(x) $ .
\item $ \forall p \in \Z, x \in \R \, : \quad h(p \cdot x) = p \cdot h(x) $ .
\item $ \forall r \in \Q, x \in \R \, : \quad h(r \cdot x) = r \cdot h(x) $ .
\item $ \forall x \in \R \, : \quad h(x) = x \cdot h(1) $ .
\end{enumerate}
In particular, the last conclusion yields that $h$ is linear.
\end{lemma*}
\begin{proof}
The first three conclusions follow from $h$ being additive. 

\paragraph{1} 
\[ h(0) = h(0) + h(x) - h(x) = h(0 + x) - h(x) = 0 \, .\]

\paragraph{2} 
\begin{align*}
    \forall x \in \R \, : \quad -h(-x) &= - \Big( h(-x) + h(x) \Big) + h(x) 
    \\
    &= - h( -x + x ) + h(x) \\
    &= - h(0) + h(x) \\
    &= h(x) \, .
\end{align*}

\paragraph{3} 
Proof by induction. The induction start $n=1$ is clear, so assume it to be true for $n \in \N$. 
\\
Then, for all $x \in \R$:
\begin{align*}
h\Big( (n+1) \cdot x \Big) &= h(n\cdot x + x) = h(n \cdot x) + h(x) \\
&= n \cdot h(x) + h(x) = (n+1)\cdot h(x) \, . 
\end{align*}

\paragraph{4} 
The statement for the case $p \in \Z \, \cap \, \{ z \geq 0 \}$ follows from the first and third conclusion. If $p \in  \Z \, \cap \, \{ z < 0 \}$, we can use the second and third conclusion to obtain for all $x \in \R$:
\begin{align*}
    h(p \cdot x) &= h\Big( (-p) \cdot (-x) \Big) = (-p) \cdot h(-x) 
    \\
    &= (-p) \cdot \Big( -h(x) \Big) = p \cdot h(x) \, .
\end{align*}

\paragraph{5} 
Write $r = \frac{p}{q}$ where $p \in \Z, q \in \N$. Then, by the fourth conclusion:
\begin{align*}
    h(r \cdot x) &= \frac{1}{q} \cdot q \cdot h\Big( \frac{p}{q} \cdot x \big) = \frac{1}{q} h\Big( q \cdot \frac{p}{q} \cdot x \big) = \frac{1}{q} h(p \cdot x ) \\
    &=\frac{1}{q} \cdot p \cdot h(x) = r \cdot h(x) \, . 
\end{align*}

\paragraph{6} 
Suppose $x \in \Q$. Then, the fifth conclusion yields
\[ h(x) = h(x \cdot 1) = x \cdot h(1) \, .\]
Therefore, suppose $x \in \R \setminus \Q$.

Since $\Q$ is dense in $\R$, we can take an increasing sequence $(r_n)_{n \in \N} \subset \Q$ that converges to $x$ (from below) and a decreasing sequence $(s_n)_{n \in \N} \subset \Q$ that converges to $x$ (from above). In the case where $h$ is an increasing function, we have for all $n \in \N$:
\begin{align*}
    r_n \leq x \leq s_n &\implies h(r_n) \leq h(x) \leq h(s_n) 
    \\
    &\implies r_n \cdot h(1) \leq h(x) \leq s_n \cdot h(1) \, .
\end{align*}
Taking the limit $n \rightarrow \infty$ in the last inequality chain yields
\[ x \cdot h(1) \leq h(x) \leq x \cdot h(1) \, .\]
If $h$ is a decreasing function instead of an increasing one, we get the same implications but with reverse inequalities in the second and last inequality chains. The end result, however, will be the same. Putting everything together yields the sixth conclusion.

\end{proof}

\begin{cor*}
\label{app:affine linearity corollary}
Let $h: \R \rightarrow \R$ be monotone and satisfy for all $z,z',\lambda \in \R$:
\begin{equation}
\label{distance preserving}
h(z+\lambda) - h(z) = h(z' + \lambda) - h(z') \, .
\end{equation} 
Then h is affine linear, i.e., there exist some $a,c \in \R$ such that for all $z \in \R\, : \, h(z) = az + c$.
\end{cor*}
\begin{proof} 
Define $h'(z):=h(z) - h(0)$, which is still a monotone function. By our assumption on $h$, we have for all $x,y \in \R$:
\begin{align*}
h'(x+y) &= h(x+y)-h(0) 
\\
&= h(x+y) - h(y) + h(y) - h(0) \\
&= h(x) - h(0) + h(y) - h(0) 
\\
&= h'(x) + h'(y) \, .
\end{align*}
Therefore, we can apply Lemma~\ref{helping analysis lemma} to $h'$ to get $a \in \R$ such that for all $z \in \R$ 
\begin{align*}
    h(z) &= h(z) - h(0) + h(0) = h'(z) + h(0) \\
    &= az + h(0) =: az + c \, .
\end{align*}
\end{proof}

\section{Utility Theory in Game Theory}
\label{sec:utility_theory}

This section revises some related utility theory and is just included for completeness. A proper treatment can be found in e.g. \citet{MasCol}.

\paragraph{Preferences and Utility Functions} Suppose a decision maker can choose one outcome from a space $C$ of $N$-many outcomes (where $N$ finite). Moreover, the decision maker prefers some outcomes over others which is captured by her preference relation $\succeq$ on $C$.

We typically describe the preferences of the decision maker through utility functions:
\begin{defn}
\label{utility repr defn}
A utility function $u : C \rightarrow \R$ is said to represent a preference relation $\succeq$ if for all $c,d \in C$, we have $c \succeq d \iff u(c) \geq u(d)$.
\end{defn}
Multiple utility functions can represent the same preference relation. Their practical use is that they translate the preference relation $\succeq$ into comparisons of numerical values.

On the other hand, starting with a utility function $u$ yields an induced preference relation $\succeq$ through
\[ \forall \, c,d \in C \quad :  \quad c \succeq d :\iff u(c) \geq u(d) \, . \] 

\paragraph{Lotteries and the Expected Utility}

Now suppose we want to allow the decision maker to choose each outcome in $C$ with some probability. Call such a probability distribution $L = (p_1,\ldots,p_N)$ over $C$ a lottery. The $i$-th outcome in C can then be represented by the lottery $e_i \in \R^n$. Thus, we extended the choice space of the decision maker from $C$ to the space $\mathcal{L}$ of lotteries. We can also extend Definition~\ref{utility repr defn} to preference relations $\succeq$ over $\mathcal{L}$ by requiring $u : \mathcal{L} \rightarrow \R$ and $\forall \, L,M \in \mathcal{L} \, : \, L \succeq M \iff u(L) \geq u(M)$.

We will be especially interested in those utility functions that simply compute the expected utility of randomly choosing an outcome according to $L$.
\begin{defn*}
    A \textit{von Neumann-Morgenstern (NM)} expected utility function is a map $U : \mathcal{L} \rightarrow \R$ that is determined by its values $U(e_i)$ on the outcomes $e_i \in C, i \in [N],$ and by 
    \[\forall L = (p_1,\ldots,p_N) \,: \quad U(L) = \sum_{i = 1}^N p_i \cdot U(e_i) \, .\]
\end{defn*}

The following theorem describes the preference relations that can be represented by a NM expected utility function. The theorem relies on four properties - called \textit{axioms} - that a preference relation $\succeq$ can satisfy: Completeness\footnote{For all $L,M \in \mathcal{L}$, we have $L \succeq M$ or $L \preceq M$ (or both, in which case we write $L \sim M$).}, Transitivity\footnote{For all $L,M,N \in \mathcal{L}$, if $L \succeq M$ and $M \succeq N$, then $L \succeq N$.}, Continuity\footnote{For all $L,M,N \in \mathcal{L}$ with $L \succeq M \succeq N$, there exists probability $p \in [0,1]$ such that $p \cdot L + (1-p) \cdot N \sim M$.} and Independence\footnote{For all $L,M,N \in \mathcal{L}$ and $p \in [0,1]$, we have $L \succeq M$ if and only if $p \cdot L + (1-p) \cdot N \succeq p \cdot M + (1-p) \cdot N$.}.
\begin{thm}[Expected Utility Theorem]
\label{exp util thm}
Let preference relation $\succeq$ satisfy the four axioms mentioned above. Then $\succeq$ can be represented by a NM expected utility function $U$. Moreover, the representing $U$ is unique up to a positive affine transformation. That is, if $U$ and $U'$ are NM expected utility functions representing $\succeq$, then there exist $\alpha,c \in \R$ such that for all $L \in \mathcal{L}$, we have $U'(L) = \alpha \cdot U(L) + c$.
\end{thm}
\begin{proof}
See Proposition~6.B.2 and 6.B.3 from \citet{MasCol}.
\end{proof}
In contrast to Theorem~\ref{exp util thm}, suppose we start with an arbitrary NM expected utility function $U$. Then $U$ induces a preference relation $\succeq$ on $\mathcal{L}$ by 
\[ \forall L, L' \in \mathcal{L} \quad :  \quad L \succeq L' :\iff U(L) \geq U(L') \, . \] 
By construction, $U$ represents $\succeq$. One can also show that this induced preference relation $\succeq$ satisfies the four axioms. Therefore, by Theorem~\ref{exp util thm}, $U$ uniquely represents the induced $\succeq$ up to a PAT.

\paragraph{Connections to Game Theory} Take a multiplayer game $G = \big( N, \{S^i\}_{i \in [N]}, \{u_i\}_{i \in [N]} \big)$. Then, the utility functions $u_i$ induce each player's preferences according to the following paragraphs:

Consider a game that only allows for pure strategy play. Then, given some player~$i$ and the pure strategy profile $s^{-i}$ of the opponents, the ``sliced'' utility function $u_i(\cdot, s^{-i})$ induces a preference relation $\succeq$ for player~$i$ over her strategy set $S^i$.

Now suppose that we allow for mixed strategy play in the games. In that case, each element in $\Delta(S^i)$ can be viewed as a lottery over the choice set $C := S^i$. Moreover, player~$i$'s utility payoff from a mixed strategy profile $\displaystyle s \in \bigtimes_{i=1}^N \Delta(S^i)$ is
\[u_i(s^i, s^{-i}) = \sum_{k_i=1}^{m_i} s_{k_i}^i \cdot u_i(e_{k_i}, s^{-i}) \, . \]
Therefore, $u_i(\cdot, s^{-i})$ has the form of a NM expected utility function. This induces a preference relation $\succeq_{i, s^{-i}}$ on the space of lotteries $\Delta(S^i)$ with $\succeq_{i, s^{-i}}$ satisfying the four axioms. Hence, $u_i(\cdot, s^{-i})$ represents the induced preference relation $\succeq_{i, s^{-i}}$ uniquely up to a PAT.

\end{document}